\documentclass[prl,twocolumn,preprintnumbers,superscriptaddress,longbibliography]{revtex4-2}
\usepackage{blochsphere}
\usepackage{graphicx}
\usepackage{tabularx}
\usetikzlibrary{backgrounds,fit,decorations.pathreplacing,calc}
\usetikzlibrary{quantikz}
\graphicspath{{images/}}
\usepackage[]{xcolor} 
\usepackage{xcolor, colortbl}
\usepackage{verbatimbox}
\usepackage{array}
\usepackage{epsfig,color}
\usepackage{tikz,lipsum,lmodern}
\usepackage[most]{tcolorbox}
\usepackage{footmisc}

\usepackage{algorithm}
\usepackage{algpseudocode}

\usepackage{physics}
\usepackage{amsmath}
\usepackage{amssymb}
\usepackage{amsthm} 
\usepackage{amsfonts}
\usepackage{mathtools} 
\usepackage{listings} 
\usepackage{thmtools} 
\usepackage{thm-restate} 
\usepackage{dcolumn}
\usepackage{bm}
\usepackage[colorlinks=true, linkcolor={urlblue},citecolor={urlblue},urlcolor={urlblue}]{hyperref}
\usepackage[sans]{dsfont}
\usepackage{capt-of}
\usepackage[normalem]{ulem}
\usepackage[caption=false]{subfig}
\newsavebox{\imagebox}
\usepackage[capitalise]{cleveref}
\usepackage{enumitem}
\usepackage[T1]{fontenc} 
\UseRawInputEncoding

\definecolor{nblue}{rgb}{0.2,0.2,0.7}
\definecolor{ngreen}{rgb}{0.2,0.6,0.2}
\definecolor{nred}{rgb}{0.7,0.2,0.2}
\definecolor{nblack}{rgb}{0,0,0}
\definecolor{urlblue}{RGB}{30,19,156}

\newcommand{\blk}{\color{nblack}}

\definecolor{mangotango}{rgb}{1.0, 0.51, 0.26}

\crefformat{appendix}{Supplemental Material~#2#1#3}
\crefname{appendix}{Supplemental Material}{Supplemental Material}

\usepackage{appendix}
\renewcommand\appendixname{Supplemental Material}

\lstdefinestyle{mystyle}{
	breaklines=true,
}
\newcommand{\ten}{\otimes}
\newcommand{\trnorm}[1]{||#1||_{\tr}} 

\def\L{\mathcal{L}}

\def\N{\mathcal{N}}

\def\E{\mathcal{E}}
\def\id{\mathcal{I}}

\def\calS{\mathcal{S}}
\def\calT{\mathcal{T}}

\DeclareRobustCommand{\upchi}{{\mathpalette\uupchi\relax}}
\newcommand{\uupchi}[2]{\raisebox{\depth}{$#1\chi$}} 

\def\brho{\bm{\rho}}
\def\bsigma{\bm{\sigma}}
\def\btau{\bm{\tau}}
\def\bchi{\bm{\upchi}}
\def\bPi{\bm{\Pi}}

\def\L{\mathbf{L}}
\def\P{\mathbf{P}}
\def\R{\mathbf{R}}
\def\I{\mathbf{I}}

\def\K{\mathbf{K}}

\def\X{\mathbf{X}}
\def\Y{\mathbf{Y}}
\def\Z{\mathbf{Z}}
\def\S{\mathbf{S}}
\def\U{\mathbf{U}}
\def\swap{\mathbf{S}}
\newcommand{\zero}{\mathbf{0}} 

\newcommand{\+}{^{\dagger}} 

\def\bea{\begin{eqnarray}}
\def\eea{\end{eqnarray}}

\def\bean{\begin{eqnarray*}}
\def\eean{\end{eqnarray*}}

\theoremstyle{definition}

\newtheorem{rst}{Result}

\newtheorem*{con*}{Conjecture}

\newcommand{\at}{atemporality}

\begin{document}
	
\title{Causal Classification of Spatiotemporal Quantum Correlations}
	\author{Minjeong Song}
	\email{song.at.qit@gmail.com}
\affiliation{Nanyang Quantum Hub, School of Physical and Mathematical Sciences, Nanyang Technological University, 637371, Singapore}
    \author{Varun Narasimhachar}
      \email{varun.achar@gmail.com }
\affiliation{Nanyang Quantum Hub, School of Physical and Mathematical Sciences, Nanyang Technological University, 637371, Singapore}
\affiliation{Institute of High Performance Computing, Agency for Science, Technology and Research (A*STAR), 1 Fusionopolis Way, Singapore 138632}
    \author{Bartosz Regula}
    \affiliation{Mathematical Quantum Information RIKEN Hakubi Research Team, RIKEN Cluster for Pioneering Research (CPR) and RIKEN Center for Quantum Computing (RQC), Wako, Saitama 351-0198, Japan}
    \author{Thomas J.~Elliott}
    \affiliation{Department of Physics \& Astronomy, University of Manchester, Manchester M13 9PL, United Kingdom}
    \affiliation{Department of Mathematics, University of Manchester, Manchester M13 9PL, United Kingdom}
    \author{Mile Gu}
    \email{mgu@quantumcomplexity.org}
\affiliation{Nanyang Quantum Hub, School of Physical and Mathematical Sciences, Nanyang Technological University, 637371, Singapore}
\affiliation{Centre for Quantum Technologies, National University of Singapore, 3 Science Drive 2, 117543, Singapore}
\affiliation{MajuLab, CNRS-UNS-NUS-NTU International Joint Research Unit, UMI 3654, Singapore 117543, Singapore}
	
	\begin{abstract}  

From correlations in measurement outcomes alone, can two otherwise isolated parties establish whether such correlations are atemporal? That is, can they rule out that they have been given the same system at two different times? Classical statistics says no, yet quantum theory disagrees. Here, we introduce the necessary and sufficient conditions by which such quantum correlations can be identified as atemporal. We demonstrate the asymmetry of atemporality under time reversal, and reveal it to be a measure of spatial quantum correlation distinct from entanglement. Our results indicate that certain quantum correlations possess an intrinsic arrow of time, and enable classification of general quantum correlations across space-time based on their (in)compatibility with various underlying causal structures.

	\end{abstract}
	
	\maketitle

Consider Alice and Bob, situated in their own laboratories. In each round, they each receive correlated random variables $A$ and $B$. The correlations were distributed via one of three possible causal mechanisms: (i) spatially such that $A$ and $B$ share a common cause, (ii) temporally such that measurement outcomes of $A$ are communicated to $B$ or vice versa, and (iii) some combination of the above (see \cref{fig:circuit}). Alice and Bob record these correlations. With only this recording, can we rule out one of the causal mechanisms above? Classical statistics says no. The observed correlations will always be compatible with all possible scenarios. Thus the adage `\emph{correlation does not imply causation}'.  

Quantum correlations can exhibit remarkable differences. Suppose Alice and Bob each receive a single qubit each round, which they then measure in some Pauli basis. The correlations between their measurement outcomes can lie outside what a density operator describes. Such \emph{aspatial correlations} cannot be explained purely by a common cause (see \cref{subfig:sdist}), leading to quantum correlations that can indeed imply causation~\cite{ried15advantage,fitzsimons15pdo}. There is thus significant interest in quantum causal inference~\cite{ried15advantage, kubler18cuasaltwoqubit, hu18discrimination, zhang20observationscheme, fitzsimons15pdo, zhao18pdogeometry,costa16causalmodelling,allen17commoncausemodel}, due to its stark departure from classical statistics. 

Here we ask, \emph{do certain quantum correlations require a common cause}? We answer in the affirmative by formalizing the notion of \emph{atemporal} quantum correlations --  correlations between Alice and Bob's qubit measurements that cannot be explained by purely temporal means (i.e., as two measurements on a qubit communicated by some quantum channel between Alice and Bob, see \cref{subfig:tdist}). We demonstrate computable necessary and sufficient indicators for atemporality and demonstrate that (1) it is asymmetric under time-reversal and thus reveals the existence of correlations possessing an intrinsic arrow of time, and (2) it represents a new operational form of non-classical correlations distinct from entanglement. This induces a framework for \emph{causal classification} -- classifying general spatiotemporal quantum correlations based on their compatibility with various causal mechanisms. Our results thus provide new mathematical tools and concepts for understanding how quantum correlations can infer causal structure in ways without classical analogs.

\medskip
\noindent\textbf{Framework.} Returning to Alice and Bob in their own laboratories, we label the qubits Alice and Bob each possess 
respectively by $A$ and $B$. We assume that the qubit pair is prepared in the same way in each round. Such a preparation scheme may be:

\begin{itemize}[itemsep=0pt, leftmargin=*, itemindent=0pt]
\item \emph{Spatially distributed}, such that $A$ and $B$ correspond to two parts of some bipartite state $\brho_{AB}$ (see \cref{subfig:sdist}).
\item \emph{Temporally distributed}, such that $B$ is the output of $A$ subject to some fixed quantum channel (completely-positive trace-preserving map) $\mathcal{E}$, or vice versa (see \cref{subfig:tdist}).
\item Neither purely spatially nor temporally distributed, such $A$ and $B$ are related by general \emph{process matrices} (see \cref{subfig:stdist}). A special case being \emph{non-Markovianity}, where evolution from $A$ to $B$ involves coupling from an ancillary environment $E$ that is initially correlated with $A$
~\cite{Note_markov,PhysRevLett.103.210401,milz21stochastic}.
\end{itemize}

\begin{figure*}[htb]
    \begin{minipage}{0.8\textwidth}
        \subfloat[Spatially distributed\label{subfig:sdist}]{\includegraphics[width=0.3\textwidth]{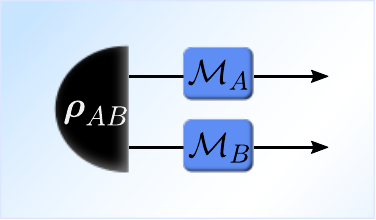}}
        \hspace{1em}
        \subfloat[Temporally distributed\label{subfig:tdist}]{\includegraphics[width=0.3\textwidth]{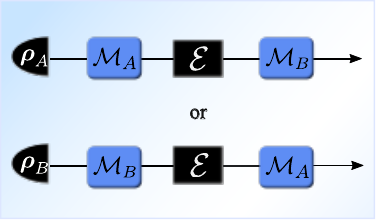}}
        \hspace{1em}
        \subfloat[General spatiotemporally distributed\label{subfig:stdist}]{\includegraphics[width=0.3\textwidth]{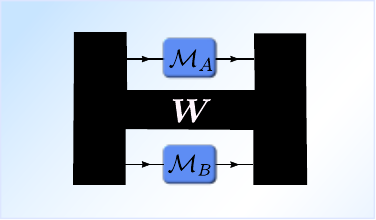}}
    \end{minipage}
    \noindent \begin{minipage}{0.12\textwidth}
    \includegraphics[width=0.8\textwidth]{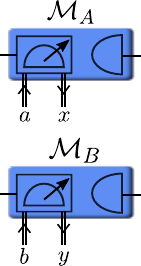}
    \end{minipage}
    \hspace{1em}
    \caption{\textbf{Causal distribution mechanisms}. Consider Alice and Bob that makes respective projective measurements $\mathcal{M}_{A}$ and $\mathcal{M}_{B}$ on respective quantum systems $A$ and $B$ (blue boxes). We divide potential mechanisms of correlating these measurements into three categories: (a) purely spatially-distributed mechanisms involving $A$ and $B$ being two aims of some bipartite state $\rho_{AB}$, reflecting the case of common cause, (b) purely temporally-distributed such that $A$ and $B$ represent the input and output of some quantum channel $\E$, reflecting the case of direct cause, or (c) some combination of both. Examples include non-Markovian evolution~\cite{PhysRevLett.103.210401,milz21stochastic} or cases of indefinite causal order \cite{oreshkov12nocausalorder}. We refer to correlations that are incompatible with (a) as aspatial and those incompatible with (b) as atemporal.} 
    \label{fig:circuit}
\end{figure*}

\noindent Let $\bsigma_{0} = \I$, $\bsigma_{1} = \X$, $\bsigma_{2} = \Y$, $\bsigma_{3} = \Z$ be the identity and the three standard Pauli operators
\footnote{Throughout this paper, we denote operators acting on Hilbert spaces of states by boldface letters.}.
Let $\Pr(x,y|a,b)$ then denote the probability of Alice getting outcome $x$ and Bob getting outcome $y$ when Alice chooses to measure in basis $\bsigma_{a}$ and Bob in $\bsigma_{b}$. Alice and Bob do not perform any other interventions. By choosing appropriate Pauli measurements over a large number of rounds, Alice and Bob can determine the expectation values
$\expval{\bsigma_{a},\bsigma_{b}} = \sum_{x,y} xy\Pr(x,y|a,b)$ describing how their measurement outcomes correlate in various Pauli basis to any desired level of accuracy. Alice and Bob then pass this information to us. What can we conclude about the causal distribution mechanisms behind the preparation of $A$ and $B$? 

Given Pauli correlations $\expval{\bsigma_{a},\bsigma_{b}}$ for each $a,b \in \{0,1,2,3\}$, Ref.~\cite{fitzsimons15pdo} proposed a concise description of this information via the pseudo-density operator (PDO)
\begin{eqnarray}
    \R_{AB} \equiv \sum_{a,b=0}^{3} \frac{\expval{\bsigma_{a},\bsigma_{b}}}{4}\bsigma_{a}\ten\bsigma_{b}.
\end{eqnarray}
Initially proposed to identify quantum correlations that imply causality, it has seen significant uses toward building quantum information theories that place space and time on equal footing~\cite{zhao18pdogeometry,pisarczyk19pdocausallimit,utagi2021pdomarkov,jia23marginal,fullwood2023pdo,marletto19pdoentanglement,marletto21temporaltele,marletto2020non,liu2024unification,liu23lighttouch,liu2024inferring}.
They contain all the information about $\expval{\bsigma_{a},\bsigma_{b}}$, since the latter can be retrieved directly via $\expval{\bsigma_{a},\bsigma_{b}} = \tr[\R_{AB} (\bsigma_a \otimes \bsigma_b)]$. Thus our capacity to infer causal distribution mechanisms from the Pauli correlations coincides with our capacity to infer causal distribution mechanisms from the corresponding PDO. Also note that when qubits $A$ and $B$ are spatially distributed, $\R_{AB}$ reduces to a standard density operator. Meanwhile, their marginal distributions are always positive and describe local measurement statistics for Alice and Bob. 

\medskip
\noindent\textbf{Spatial and temporal compatibility}. We introduce two distinct criteria on PDOs: We say that $\R_{AB}$ is \emph{spatially compatible}, or belongs to $\calS$ if its statistics can be generated via a spatial distribution mechanism (i.e., as in \cref{subfig:sdist}). Similarly, we say that $\R_{AB}$ is \emph{temporally compatible}, or belongs to $\calT$ if its statistics can be generated via a temporal distribution mechanism (i.e., as in \cref{subfig:tdist}). We will often use the terms \emph{spatial} and \emph{temporal} for brevity, but we stress that they only mean \emph{compatible with} a spatially or temporally distributed structure. PDOs that lie outside of $\calS$ are referred to as \emph{aspatial}, and those that lie outside of $\calT$ are referred to as \emph{atemporal}.

We then divide the set of all PDOs using a Venn diagram into four separate classes based on their spatial-temporal compatibility: Those that (a) lie in $\calS$ and $\calT$ and are thus compatible with any distribution mechanism, (b) lie in $\calS$ but not $\calT$ and thus rule out purely temporal distribution mechanisms, (c) lie in $\calT$ but not $\calS$ and thus rule out purely spatial distribution mechanisms and (d) those that lie outside $\calS$ and $\calT$ that cannot be explained by either purely spatial or temporal distribution scheme but rather, rely on a more complicated combination of spatial and temporal mechanisms. States that lie in (a) behave like classical probability distributions, and we cannot infer anything conclusive about their underlying causal distribution mechanism. Quantum correlations, however, permit PDOs in each of (b), (c), and (d), where certain causal distribution mechanisms can be ruled out (see \cref{fig:Venn}).

\begin{figure}[ht]
    \includegraphics[width=0.48\textwidth]{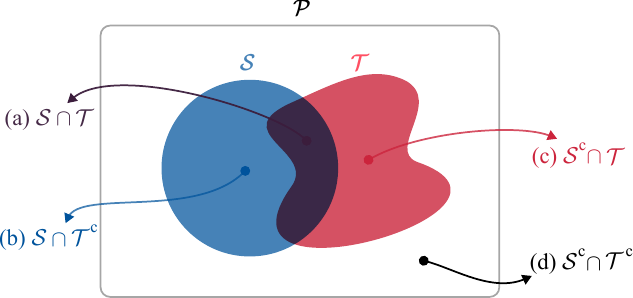}
    \caption{\textbf{Venn diagram of all spatial--temporal compatibility}. The set $\mathcal P$ of observed spatiotemporal quantum correlations (as described by PDOs) is divided into four mutually exclusive subsets: (a) $\calS \cap \calT$ represent correlations that are compatible with purely spatial and purely temporal distribution mechanisms, 
    (b) $\calS \cap \calT^c$ represents correlations that rule out purely temporal distribution mechanisms,
    (c) $\calS^c \cap \calT$ represents correlations that rule out purely spatial distribution mechanisms such as two coexisting qubits measured separately, and (d) $\calS^c \cap \calT^c$ designates correlations that require a combination of spatial and temporal distribution mechanisms to explain. 
    Note that unlike $\mathcal{S}$, $\mathcal{T}$ does not form a convex set (see Example 5 in \appendixname).}

    \label{fig:Venn}
\end{figure}

To better understand what PDOs lie within each class, we need a necessary and sufficient criterion for aspatiality and atemporality. PDOs were initially introduced to study the former, with Ref.~\cite{fitzsimons15pdo} showing that the negativity of $\R_{AB}$ is necessary and sufficient for \emph{aspatiality}. Some works also looked into temporal correlations in some limited scenarios under additional assumptions, e.g., maximally mixed or full ranked initial state \cite{ried15advantage,kubler18cuasaltwoqubit,zhao18pdogeometry, fullwood2023pdo}, or unitary evolution \cite{hu18discrimination,zhang20observationscheme}. We will derive conditions without such assumptions for when a PDO is \emph{atemporal}, and thus build a full picture of spatial-temporal compatibility.

\medskip
\noindent\textbf{Certifying atemporality}. Given $\R_{AB}$, our goal is to determine identifiers of atemporality that rule out compatibility with temporal distribution mechanisms. Consider first a \emph{forward atemporality} measure $\overrightarrow{f}$ that is zero if and only if $\mathbf{R}_{AB}$ has statistics consistent with a temporal distribution mechanism from $A$ to $B$; and a \emph{reverse atemporality} measure $\overleftarrow{f}$ that is zero if and only if $\mathbf{R}_{AB}$ has statistics consistent with a temporal distribution mechanism from $B$ to $A$. Together they naturally induce a general \emph{atemporality} measure $f = \mathrm{min}(\overrightarrow{f},\overleftarrow{f})$ that is zero if and only if $\R_{AB}$ lies in $\mathcal{T}$.

We then introduce \emph{pseudo-channels}, a temporal analog of pseudo-density operators. Recall that the Choi--Jamio{\l}kowski isomorphism enables us to represent each qubit channel $\Lambda$ by a Choi operator~\cite{choi75}
\begin{eqnarray}
    \bchi_{\Lambda} \equiv (\id \ten \Lambda) \ketbra*{\phi^{+}},
\end{eqnarray}
describing the output state when $\Lambda$ is applied to one arm of a Bell state $\ket*{\phi^{+}} = \frac{1}{\sqrt{2}}(\ket{00}+\ket{11})$. Here, $\id$ denotes the identity channel. We observe that this output does not need to be a valid spatial quantum state if $\Lambda$ is not a valid quantum channel. More generally, let $\bchi_{\Lambda}$ be a PDO with non-zero negativity $\mathcal{N}(\bchi_{\Lambda})$ (absolute sum of its negative eigenvalues). In such scenarios, $\Lambda$ remains trace-preserving, Hermiticity-preserving, and linear, but is no longer completely positive.

\begin{figure}[ht]
    \centering
    \includegraphics[width=0.95\linewidth]{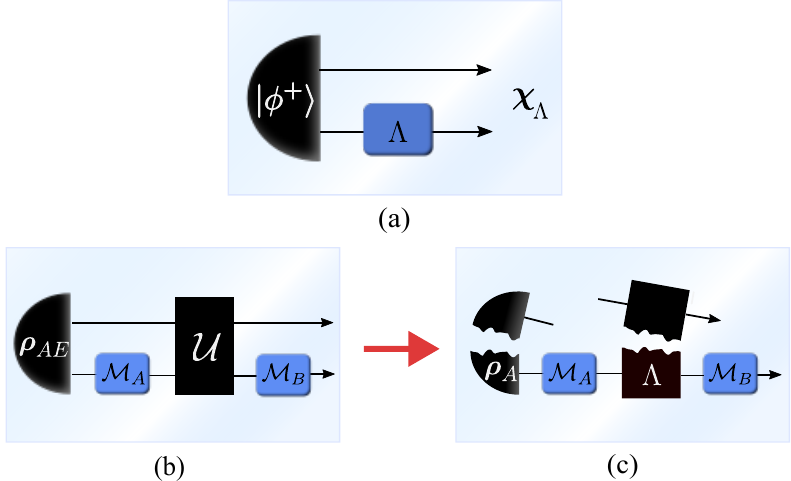}
    \caption{\textbf{Choi operator and pseudo-channels}. (a) The Choi operator of a quantum channel $\Lambda$ describes the resulting state when we apply $\Lambda$ on one arm of the Bell state $\ket{\phi^{+}}\equiv \frac{1}{\sqrt{2}}(\ket{00}+\ket{11})$.
    (b) In cases where the resulting correlations $\R_{AB}$ between $A$ and $B$ are atemporal, there is no valid quantum channel from $A$ to $B$. (c) Nevertheless, we can identify a pseudo-channel $\Lambda$ that forcibly interprets the dynamics as a map from $A$ to $B$. The resulting $\Lambda$ is non-physical, which is reflected by the negativity of its Choi operator. We show that this negativity is a necessary and sufficient condition for the forward atemporality of $\R_{AB}$.}
    \label{fig:choi_pc}
\end{figure}

This motivates us to define \emph{pseudo-channels}: linear maps that preserve trace and hermiticity but which can be non completely positive. $\Lambda$ is then a \emph{pseudo-channel}, and $\mathcal{N}(\bchi_{\Lambda})$ provides a necessary and sufficient indicator of its \emph{non-physicality}. Such pseudo-channels provide a natural means to define $\overrightarrow{f}$ (and thus $\overleftarrow{f}$ and $f$). Given $\R_{AB}$, we first assert that it describes correlations resulting from some quantum channel $\overrightarrow{\Lambda}$ with input system $A$ and output system $B$. Ref.~\cite{zhao18pdogeometry,horsman17pdoandothers} demonstrated that for temporal PDOs, the associated quantum channel satisfies 
    \begin{equation}
            \R_{AB} = \left(\id_A \ten \overrightarrow{\Lambda} \right) \K_{AB} \label{eq:temporal_PDO}
    \end{equation}
\noindent where $\K_{AB} \equiv \left\{ \brho_A \ten \frac{\I_B}{2}, \S_{AB} \right\}$ with $\brho_A \equiv \tr_B \R_{AB}$ being the first marginal, $\{\cdot,\cdot\}$ the anti-commutator, $\I$ the identity operator, and $\S$ the swap operator. In more general cases where $\R_{AB}$ is not necessarily a temporal PDO, we rationalize the following;
When $\R_{AB}$ is incompatible with a temporal distribution mechanism from $A$ to $B$, no such valid quantum channels exist. However, we can drop the complete-positivity requirements on $\overrightarrow{\Lambda}$. In \appendixname~\footnote{See Supplemental Material.} (see Lemma 1), we prove that any PDO will have at least one compatible forward pseudo-channel. This allows us to interpret \emph{any} spatiotemporal correlations as resulting from a pseudo-channel $\overrightarrow{\Lambda}$ acting on $A$ to generate $B$ (see \cref{fig:choi_pc}). The minimal non-physicality of such a channel then motivates our definition for \emph{forward atemporality}:
\begin{equation}
\overrightarrow{f}(\R_{AB})\equiv\min_{ \overrightarrow{\Lambda}} ~\mathcal{N}(\bchi_{\overrightarrow{\Lambda}}),
\end{equation}
\noindent where the minimization is over all forward pseudo-channels $\overrightarrow{\Lambda}$ that is compatible with $\R_{AB}$. Similarly, we define the \emph{reverse atemporality} $\overleftarrow{f}$ by interchanging $A$ and $B$, thus also defining the overall atemporality $f = \mathrm{min}(\overrightarrow{f},\overleftarrow{f})$. For example, the entangled Bell state $\frac{1}{\sqrt{2}}(\ket{01} - \ket{10})$ has forward, reverse and overall atemporality of $0.5$ (its corresponding pseudo-channel being the non-physical universal-NOT gate~\cite{bu2000universal}). This then leads to one of our key results:

\begin{rst}     
Given spatiotemporal correlations described by a PDO $\R_{AB}$, let
    \begin{equation}
\overrightarrow{f}(\R_{AB})\equiv\min_{ \overrightarrow{\Lambda}} ~\mathcal{N}(\bchi_{\overrightarrow{\Lambda}}),
\end{equation}
    where the minimization is over all forward pseudo-channels $\overrightarrow{\Lambda}$ that is compatible with $\R_{AB}$ (i.e., those that satisfy Eq. (\ref{eq:temporal_PDO})). Then $f > 0$ is a necessary and sufficient condition for atemporality. Moreover, we have a systematic algorithm (see below) to compute $\overrightarrow{f}$ for any $\R_{AB}$.\label{rst:1}
\end{rst}

Observe first $f > 0$ implies that no physical channel is compatible with $\R_{AB}$ by definition, while existence of a compatible $\overrightarrow{\Lambda}$ (as expressed by their Choi operator) can be systematically identified as follows: When $\R_{AB}$ has full rank marginals, $\overrightarrow{\Lambda}$ is unique and the algorithm that returns its Choi operator $\bchi$ is particularly simple (see \cref{alg:pseudo_channel}). From this, the forward atemporality can be directly computed. 

 \begin{algorithm}[H]
\caption{Choi operator of pseudo-channel construction}\label{alg:pseudo_channel}
\begin{algorithmic}[1]
\Require $2$-qubit PDO $\R_{AB}$
\State $\brho_A \gets \tr_B \R_{AB}$
    \State $\L \gets \left(\brho_A-\frac{\I}{2}\right)\ten\tr_A\left[\left(\frac{1}{2}\brho_A^{-1} \ten \I\right) \R_{AB}\right]+\frac{\I}{2}\ten\tr_A\left[\left((\I-\frac{1}{2}\brho_A^{-1}) \ten \I\right) \R_{AB}\right]$

    \State $\bchi \gets (T \ten \id)(\R_{AB} - \L)$
    \color{gray}{\Comment{ $T$ denotes the transpose map}} \blk
    
\State \Return{$\bchi$}
\end{algorithmic}
\end{algorithm}

When the first marginal $\brho_A$ is rank-deficient (i.e., some pure state $\ketbra{\phi}$), the pseudo-channels compatible with $\R_{AB}$ are no longer unique. This is because any such causal interpretation corresponds to Alice being given a system in $\ket{\phi}$, such that measured statistics do not contain any information regarding outputs when $\Lambda$ acts on a state $\ket{\smash{\phi^{\perp}}}$ perpendicular to $\ket{\phi}$. In \appendixname~\cite{Note2}, we generalize the above algorithm to identify all compatible pseudo-channels, and a semi-definite program to find the minimum non-physicality among them. Thus, we can systematically evaluate the forward atemporality $\overrightarrow{f}$ for all $2$-qubit PDOs. Interchanging $A$ and $B$ enables evaluation of the reverse atemporality $\overleftarrow{f}$ and thus overall atemporality $f$.

\medskip
\noindent\textbf{Properties of atemporality}. The computability of atemporality offers efficient means to study its properties. Here, we survey key results (see \appendixname~\cite{Note2} for further details). The first is time-reversal asymmetry. Unlike classical correlations, certain quantum correlations admit temporal mechanisms in only one temporal direction.

\begin{rst} Forward atemporality does not imply reverse atemporality or vice versa.
\end{rst}

Consider a PDO describing a single qubit $A$ undergoing probabilistic dephasing $\E(\brho):=p\brho+(1-p)\Z\brho\Z\+$, the output of which we label qubit $B$. Clearly, its forward atemporality is $0$ by construction. However, $\overleftarrow{f}(\R_{AB}) > 0$ for all $p \neq 0,1$ (see Example 8 in \appendixname~\cite{Note2}). Thus, quantum correlations not only can imply causality as previously suggested~\cite{fitzsimons15pdo, ried15advantage}, but can also only imply causality in a particular temporal direction.

The interplay of temporal and spatial compatibility is also a natural point of interest. Specifically, let us restrict ourselves to density operators (i.e., correlations that lie in $\mathcal{S}$) and let $\{\proj{e_j}\}$ and $\{\proj{f_j}\}$ be some orthonormal basis respectively on $A$ and $B$. In \appendixname~\cite{Note2}, we show that classical distributions of such orthogonal states, i.e., of form $\sum_{jk} p_{jk} \proj{e_j} \ten \proj{f_k}$ have zero atemporality in either direction -- aligning with the intuition that classical statistics cannot rule out any causal distribution mechanism without active intervention. Meanwhile states of the form $\rho_{AB} = \sum_{j} p_j \proj{e_j} \ten \btau_j$, where $\btau_j$ is an arbitrary state on $B$ (i.e., those with zero one-way discord~\cite{ollivier2001quantum,henderson2001classical,modi2012classical}) has zero forward atemporality. 

One might also speculate that entanglement implies atemporality. Indeed, we show whenever $\R \in \calS$ is pure or has maximally mixed marginals (see Theorem 3 in \appendixname~\cite{Note2}), non-zero atemporality coincides with non-zero entanglement. However, this does not hold in more general conditions:
\begin{rst}
   Entanglement does not imply atemporality: Certain entangled states are temporally compatible.  \label{obs:first}
\end{rst}
Consider the parameterized family of biased Werner states  $\brho_{p,q}^\text{B.W.} \equiv (1-p)\brho^\text{Werner}_{q} + p\proj{00}$, achieved by mixing a standard Werner state $\brho^\text{Werner}_{q}$ \cite{werner89}
with the state $\ketbra{00}$
(see \cref{fig:bwerner}). Here, $\brho_{0.5,0.25}^\text{B.W.}$ has zero atemporality, but non-zero entanglement negativity ($\approx0.0087$) \cite{vidal02negativity}. Nevertheless, sufficiently strong entanglement does guarantee atemporality. In \appendixname~\cite{Note2} (see Theorem 4 within), we prove the following:
\begin{rst} Any temporally compatible $2$-qubit state must have entanglement negativity of at most $\frac{1}{2}(\sqrt{2}-1)$.
\end{rst}
\noindent Indeed a scatter plot of atemporality vs. entanglement negativity for 1000 randomly generated density operators suggest that two concepts are heavily correlated but not the same -- with atemporality looking to be a stronger notion of non-classical correlations than entanglement (see \appendixname~\cite{Note2}). Thus we anticipate that future study of atemporality could well lead to a new and finer-grained understanding of quantum correlations.

\begin{figure}[tb]
    \centering
    \includegraphics[width=0.95\linewidth]{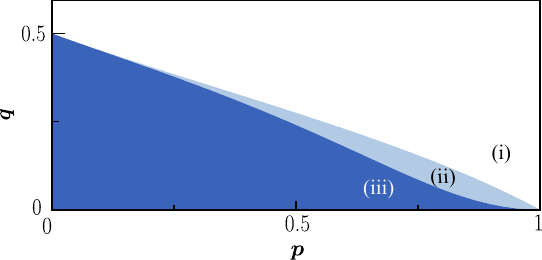}
    \caption{ \textbf{Entanglement and atemporality.} The family of biased Werner states $\brho_{p,q}^\text{B.W.} \equiv (1-p)\brho^\text{Werner}_{q} + p\proj{00}$ illustrates the differences between entanglement and atemporality. The plot depicts a colour map depicting entanglement and atemporality of $\brho_{p,q}^{B.W.}$ for various values of $p$ and $q$. While (i) all separable states are not atemporal, (ii) there exist entangled states that nevertheless admit temporal distribution mechanisms. Notwithstanding, (iii) most entangled states within the family are atemporal. 
    }
    \label{fig:bwerner}
\end{figure}

\medskip
\noindent\textbf{Discussion}.
Spatiotemporal quantum correlations differ crucially from classical counterparts in that they can be fundamentally incompatible with certain underlying causal distribution mechanisms. In this article, we showed that such correlations between various Pauli measurements on two qubits $A$ and $B$ can be \emph{atemporal}, such that their explanation necessitates some common cause. We provided a necessary and sufficient indicator of atemporality and a systematic algorithm to compute it. In studying atemporality, we illustrated (1) the existence of temporal asymmetry -- whereby certain correlations admit purely causal explanations in only one temporal direction and (2) that \emph{atemporality} induces a notion of quantum correlations distinct from entanglement. Combined with prior work showing quantum correlations can also be \emph{aspatial} -- such that they cannot be purely explained by a common cause --  our results enable a framework to classify quantum correlations based on their compatibility with \emph{spatial} and \emph{temporal} mechanisms. 

This classification opens a number of interesting directions. Our work here focused on the spatiotemporal correlation between two qubits as it allowed for closed-form expressions for atemporality; however, the fundamental concept introduced applies to arbitrary bipartite systems. Fundamentally, we can define the atemporality of any correlations between $A$ and $B$ (or vice versa) as how non-physical a quantum channel from $A$ and $B$ (or vice versa) must be to generate the correlations observed. Indeed, PDOs are well-defined for bipartite systems with $n$-qubit partitions~\cite{liu2024unification}, while variants extend these ideas to general $d$-dimensional or continuous variable systems~\cite{fullwood2024operator,zhang20different}. The identification of an analog of our \cref{alg:pseudo_channel} for finding (pseudo)-channels then provides a natural pathway for understanding spatiotemporal compatibility on systems of arbitrary dimensions. Meanwhile, our quantifier of atemporality used one particular measure of non-physicality. This choice is not unique; other definitions of non-physical maps and quantifiers of non-physicality exist~\cite{dominy16beyondcp, pechukas94reduced,paz19dynamics,lu2016structure,modi12reduced, buscemi14reduced, vacchini16reduced, utagi2021pdomarkov, parzygnat21conditional, parzygnat22non,parzygnat23bayes,parzygnat23axioms,regula2021operational, carteret2008dynamics}; and thus we can anticipate many alternative measures of atemporality akin to what exists for entanglement. 

The relation between atemporality and spatial quantum correlations also yields fascinating insights. We proved that some forms of quantum correlation are required for a standard bipartite quantum system to be temporal, and sufficient entanglement guarantees atemporality.  Still, many open questions remain. Some entangled states are not atemporal, so is atemporality a strictly stronger notion of quantum correlations? And if so, is it guaranteed by steering or Bell non-locality~\cite{wiseman2007steering}? We also have no proof that atemporality guarantees entanglement, and thus could atemporality persist in more robust forms of quantum correlations~\cite{modi2012classical}? Whatever the case, atemporality has a clear operational interpretation and introduces an entirely new category to the existing hierarchy of quantum correlations, and answering such questions will help us better understand the uniquely quantum incompatibility between spatial and temporal correlations.

\medskip
\begin{acknowledgments}
\noindent\textbf{Acknowledgements}.
This work is supported by the National Research Foundation, Singapore, and Agency for Science, Technology and Research (A*STAR) under its QEP2.0 programme (NRF2021-QEP2-02-P06) and its CQT Bridging Grant, the Singapore Ministry of Education Tier 1 Grants RG77/22 and RT4/23, the Singapore Ministry of Education Tier 2 Grant MOE-T2EP50221-0005, grant no.~FQXi-RFP-1809 (The Role of Quantum Effects in Simplifying Quantum Agents) from the Foundational Questions Institute and Fetzer Franklin Fund (a donor-advised fund of Silicon Valley Community Foundation). TJE is supported by the University of Manchester Dame Kathleen Ollerenshaw Fellowship. VN acknowledges support from the Lee Kuan Yew Endowment Fund (Postdoctoral Fellowship).  
\end{acknowledgments}

\bibliography{main.bib}

\begin{thebibliography}{49}
\bibitem [{Not()}]{Note_markov}%
  \BibitemOpen
  \href@noop {} {}\bibinfo {note} {A general quantum process describes evolution of an open system in contact with an external enviroment. Markovian evolution then aligns with the special case where optimal modelling of the the systems future behaviour requires only knowledge of its current state $\rho_0$, whereby its future state $\rho_t$ can be described by quantum channels (completely positive and trace-preserving maps) acting on $\rho_0$. A non-Markovian process comes about when one can gain more information about $\rho_t$ by knowing the state of the enviroment - as this signifies that information about the system-enviroment correlations established at earlier times ($t < 0$) can be fed back into the system. A signature of non-Markovianty is thus when the mapping between $\rho_0$ and $\rho_t$ cannot be described by a quantum channel}\BibitemShut {NoStop}%
\end{thebibliography}

\onecolumngrid
\begin{appendix}
    \section{Supplemental Material}
    \subsection{Examples}\label{app1}

Recall that we focus on two-level bipartite quantum systems, and that we often use the terms spatial (respectively, temporal) for brevity, but they only mean compatible with a spatially (resp. temporally) distributed structure. $\calS$ denotes the set of all spatial pseudo density operators (PDOs) and $\calT$ denotes that of all temporal PDOs. Here, we list various examples for each subsets $\mathcal{S}\cap\mathcal{T}^c$, $\mathcal{S}^c\cap\mathcal{T}$, $\mathcal{S}\cap \mathcal{T}$, and $\mathcal{S}^c \cap \mathcal{T}^c$.

\begin{restatable}{exm}{exsat}
    Maximally entangled states are spatial but atemporal.
\end{restatable}

\begin{proof}
    Maximally entangled states are represented by density operators so it is trivial that they belong to $\mathcal{S}$. Observe that the marginals of maximally entangled states are the maximally mixed state. By using \cref{thm:atempent}, we observe nonzero atemporalities, which implies that maximally entangled states do not belong to $\mathcal{T}$.\\
    
    Hence, we conclude that maximally entangled states belong to $\mathcal{S}\cap\mathcal{T}^c$.
\end{proof}

\begin{restatable}{exm}{exast}
    Let $\R_{AB}$ be the PDO constructed from temporally distributed quantum systems, $A$ and $B$, with its associated map being the noiseless quantum channel $\id$ from $A$ to $B$. Then, the PDO $\R_{AB}$ is temporal but aspatial.
\end{restatable}

\begin{proof}
    It is trivial that the PDO $\R_{AB}$ belongs to $\mathcal{T}$, so we only need to show that $\R_{AB} \notin \mathcal{S}$. We will prove this by showing that $\R_{AB}$ has a negative eigenvalue. Following Zhao et al.~\cite{zhao18pdogeometry} and Horsman et al.~\cite{horsman17pdoandothers}, with an arbitrary initial state $\brho_A \equiv \dfrac{1}{2}\begin{bmatrix}
        1+p & c \\
        c^* & 1-p
    \end{bmatrix}$ where $p \in \mathbb{R}$ and $c \in \mathbb{C}$ and the noiseless quantum channel $\id$ such that $\id(\X) = \X$ for any linear operators $\X$, the PDO $\R_{AB}$ has its matrix form in the canonical basis as follows, 
    \begin{eqnarray}
        \R_{AB} &=& (\id \ten \id) \left\{ \brho_A \ten \frac{\I_{B}}{2}, \S \right\}\\
        &=& \frac{1}{4}
        \begin{bmatrix} 
            2+2p & c & c & 0 \\
            c^* & 0 & 2 & c \\
            c^* & 2 & 0 & c \\
            0 & c^* & c^* & 2-2p
        \end{bmatrix},
    \end{eqnarray}
    where $(\cdot)^*$ denotes the complex conjugate. Note that $\S$ denotes the swap operator, i.e., $\S = \frac{1}{2} \sum_{k=0}^3 \bsigma_k \ten \bsigma_k$, or equivalently $\S = \sum_{i,j=0}^1 \ketbra{i}{j}\ten\ketbra{j}{i}$. To gain the eigenvalues of $\R_{AB}$, we compute the characteristic equation $\det(\R_{AB}-\lambda\I)=0$:
    \begin{eqnarray}
        0 &=& \det(\R_{AB}-\lambda\I) \\
        &=& \det \frac{1}{4} \begin{bmatrix} 
            2+2p-4\lambda & c & c & 0 \\
            c^* & -4\lambda & 2 & c \\
            c^* & 2 & -4\lambda & c \\
            0 & c^* & c^* & 2-2p-4\lambda
        \end{bmatrix}\\
        &=&\det \frac{1}{4} \begin{bmatrix} 
            2+2p-4\lambda & c & c & 0 \\
            c^* & -4\lambda & 2 & c \\
            0 & 2+4\lambda & -4\lambda-2 & 0 \\
            0 & c^* & c^* & 2-2p-4\lambda
        \end{bmatrix}\\
        &=&\left(\lambda-\frac{1}{2}\right)\left(\lambda+\frac{1}{2}\right)\left(\lambda^2-\lambda-\frac{1}{4}\left(\abs{p}^2+\abs{c}^2-1\right)\right).
    \end{eqnarray} 
    We obtain four eigenvalues of $\R_{AB}$; $-\frac{1}{2},\frac{1}{2},\frac{1}{2}(1\pm\sqrt{|p|^2+|c|^2})$, and observe that it always has at least one negative eigenvalue $-\frac{1}{2}$. In fact, it has only one negative value because the positivity of the initial state leads $|p|^2+|c|^2 \le 1$ and thus $\frac{1}{2}\left(1\pm\sqrt{|p|^2+|c|^2}\right)$ are always nonnegative.\\
    
    Hence, we conclude that $\R_{AB}$ belongs to $\mathcal{S}^c\cap\mathcal{T}$ for any initial states.
\end{proof}

\begin{restatable}{exm}{exprod}
    Product states are both spatial and temporal.
\end{restatable}

\begin{proof}
    As product states are density operators, it is left to show that product states belongs to $\mathcal{T}$. Let $\brho \ten \btau$ be a product state for some density operators $\brho,\btau$. Now consider temporally distributed quantum systems of which the initial state is given by $\brho$ and the associated quantum channel is given by the constant channel $\E$ such that $\E(\cdot) = \tr[\cdot] \btau$. By using the property of temporal PDOs as in Ref.~\cite{zhao18pdogeometry,horsman17pdoandothers}, it is easily shown that the PDO $\R$ representing these temporally distributed quantum systems, coincides with $\brho \ten \btau$:
    \begin{eqnarray}
        \R &=& (\id \ten \E) \left\{ \brho \ten \frac{\I}{2}, \S \right\}\\
        &=& \frac{1}{4} \sum_k \left\{ \brho, \bsigma_k \right\} \ten \E(\bsigma_k) \\
        &=& \frac{1}{4} \sum_k \left\{ \brho, \bsigma_k \right\} \ten 2\delta_{k,0}\btau \\
        &=& \brho \ten \btau. 
    \end{eqnarray}

Hence, we conclude that $\brho \ten \btau$ belong to $\mathcal{S}\cap\mathcal{T}$.
\end{proof}

\begin{restatable}{exm}{exdiscord}
    Let $\R_{AB}$ be a separable quantum state such that $\R_{AB} = \sum_{i=0}^{1} p_i \proj{e_i} \ten \btau_i$, where $\sum_i p_i =1$, $p_i \ge 0$, and $\{\ket{e_i}\}_i$ is an orthonormal basis of Hilbert space of the system $A$. Then, $\R_{AB}$ is both spatial and temporal. \label{ex:discord}
\end{restatable}

\begin{proof}
    Consider a completely positive and trace-preserving (CPTP) map $\E: \brho \mapsto \sum_{i=0}^{1}\tr\left[\brho \proj{e_i}\right] \btau_i$. Now we show that $\E$ is a pseudo-channel compatible with $\R_{AB}= \sum_{i=0}^{1} p_i \proj{e_i} \ten \btau_i$. To this end, it suffices to show that $\E$ satisfies 
    \begin{equation}
        \R_{AB} = \frac{1}{4} \sum_k\left\{ \brho_A, \bsigma_k \right\} \ten \E(\bsigma_k), 
    \end{equation} 
    where $\brho_A \equiv \tr_B \R_{AB}$, i.e., $ \brho_A = \sum_i p_i \proj{e_i}$.

    Indeed,
    \begin{eqnarray}
        \frac{1}{4} \sum_k\left\{ \brho_A, \bsigma_k \right\} \ten \E(\bsigma_k) &=& \frac{1}{4} \sum_k\left\{ \brho_A, \bsigma_k \right\} \ten \sum_i\tr\bigl[\bsigma_k \proj{e_i}\bigr] \btau_i\\
        &=& \frac{1}{2} \sum_{i}\left\{ \brho_A, \biggl(\frac{1}{2}\sum_{k}\tr\bigl[\bsigma_k \proj{e_i}\bigr]\bsigma_k\biggr) \right\} \ten \btau_i\\
        &=& \frac{1}{2} \sum_{i}\bigl\{ \brho_A,  \proj{e_i} \bigr\} \ten \btau_i\\
        &=& \frac{1}{2} \sum_{i}\left\{ \sum_j p_j \proj{e_j},  \proj{e_i} \right\} \ten \btau_i\\
        &=& \frac{1}{2} \sum_{i,j} p_j \delta_{i,j}\left( \ketbra{e_j}{e_i}+ \ketbra{e_i}{e_j} \right) \ten \btau_i\\
        &=&\frac{1}{2} \sum_{i} 2p_i \proj{e_i} \ten \btau_i\\
        &=& \sum_{i} p_i \proj{e_i} \ten \btau_i = \R_{AB}.
    \end{eqnarray}
    The third equation holds due to the fact that Pauli operators and the identity operator form a complete basis for all linear operators and so $\proj{e_i} = \frac{1}{2}\sum_{k}\tr[\bsigma_k \proj{e_i}]\bsigma_k, \forall i$.\\
    
    Hence, we conclude that any separable states of zero discord \cite{ollivier2001quantum,henderson2001classical,modi2012classical} belong to $\mathcal{S}\cap\mathcal{T}$.
\end{proof}
It is noteworthy that \cref{ex:discord} is consistent with the known result in Ref.\ \cite{brodutch13vanishingdiscord}. Imagine that the initial state between the system of interest $A$ and an environmental system $E$ is prepared in a zero discord state as above, the composite system after Alice's observation is subject to the swap operation, and Bob receive the system $E$ instead of the system $A$. Then the collective state of the composite system $AB$ coincides with the initial zero discord state. In such manner, the collective state of the composite system $AB$ can be interpreted as that of the composite system $AE$. In proof, we showed that the zero discord of the initial state induces the reduced dynamic (as represented by a pseudo-channel) which is always represented by a CPTP map. This provides a comprehensive understanding how vanishing quantum discord in the initial correlation between the system and the environmental system is related to complete positivity of reduced dynamics, by constructing the reduced dynamics explicitly. 

Example.~\ref{ex:discord} showed that atemporality implies quantum discord, and we will see later with Example.~\ref{mex:atem_ent} that the converse is false.

\begin{restatable}{exm}{excounter}
    Consider two forward-temporal processes : (i) the initial system is prepared in the state $\ket{0}$ in the Z-basis and the intervening channel is given by the noiseless channel $\id$. We denote its PDO by $\R_{\ket{0},\id}$. (ii) the initial system is prepared in the state $\ket{+} \equiv \frac{1}{\sqrt{2}}(\ket{0}+\ket{1})$ and the intervening channel is given by a noise quantum channel $\N$ such that $\N(\brho) = \tr[\brho] \frac{\I}{2}$. We denote its PDO by $\R_{\ket{+},\N}$. Then, a probabilistic mixture of these two temporal PDOs with probability $p=\frac{1}{2}$, i.e. $\R\equiv\frac{1}{2}(\R_{\ket{0},\id}+\R_{\ket{+},\N})$, is no longer temporal as we observe from numerical data that $f(\R) \approx 0.0785$. This is a counter example for the set $\mathcal{T}$ being convex. Note that $\R$ is also not spatial either, as it has a negative eigenvalue. Thus, mixing two forward-temporal (or two reverse-temporal) PDOs can already result in a PDO in $\mathcal{S}^c \cap \mathcal{T}^c$.\label{ex:counter}
\end{restatable}

\begin{restatable}{exm}{exwer}
    Given a Werner state $\R_{AB} \in \mathcal{S}$, $\R_{AB}$ is separable if and only if $\R_{AB}$ is temporal.
    \label{mex:Wer}
\end{restatable}

	\begin{proof}
		Werner states are one of the well-studied classes of quantum states, and their mathematical definition was explicitly given in Ref.~\cite{werner89}. A two-qubit Werner state $\R_q$ is parameterized by $q \in [0,1]$ as 
		\begin{equation}
			\R_q = \frac{q}{3}\P_s + (1-q)\P_a, \label{eq:werner}
		\end{equation} 
		where $\P_s,\P_a$ are the projectors onto the symmetric space, the antisymmetric space, respectively. That is, 
		\begin{eqnarray}
			\P_s &=& \frac{1}{2}\bigl(\I\ten\I + \S\bigr), \\
			\P_a &=& \frac{1}{2}\bigl(\I\ten\I - \S\bigr),
		\end{eqnarray}
		where $\S=\dfrac{1}{2}\sum\limits_{i=0}^{3} \bsigma_i \ten \bsigma_i$.
		As marginals of a Werner state are always the maximally mixed state, we observe that $E_\text{neg}(\R_q) = f(\R_q)$ for all Werner states by using \cref{thm:atempent}.\\
  
        Hence, we conclude that  a Werner state $\R_q$ is separable if and only if $\R_q\in \mathcal{T}$. 
	\end{proof}

\cref{fig:werner} illustrates the identity of entanglement and atemporality in the class of Werner states. We observe that not only the presence or absence of atemporality and atemporality coincide, but their quantities also coincide. 
\begin{figure}[t]
    \centering
    \includegraphics[scale=0.6]{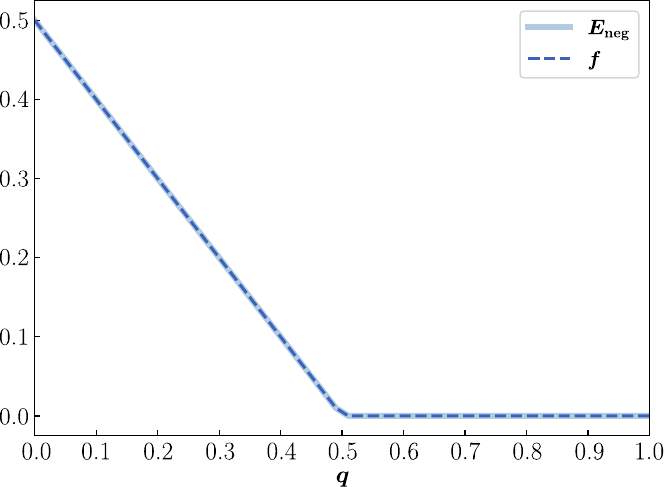}
    \caption{The graph of the entanglement negativity $E_\text{neg}$ of Werner states $\brho_q^\text{Werner}$ parameterized by $q \in [0,1]$ against the parameter $q$ is plotted by the light blue solid line, and the graph of the \at, $f$ is plotted by the blue dashed line. }
    \label{fig:werner}
\end{figure}

\begin{restatable}{exm}{exunitary}
    Let $\R_{AB}$ be a forward temporal PDO,
    and $\brho_A$ the associated initial state.
    If a unitary channel $\mathcal{U}$ is a forward pseudo-channel for $\R_{AB}$,
    then $\R_{AB}$ is reverse temporal with the adjoint channel $\mathcal{U}^\dagger$
    as its reverse pseudo-channel and $\mathcal{U}(\brho_A)$ as its associated initial state.
\end{restatable}

    \begin{proof}
        Since $\R_{AB}$ is a forward temporal PDO with its first marginal being $\brho_A$ and its pseudo-channel being $\mathcal{U}$, following Ref.~\cite{zhao18pdogeometry,horsman17pdoandothers}, we have
        \begin{equation}
        \R_{AB} = \left(\id \ten \mathcal{U}\right) \left\{ \brho_A \ten \frac{\I}{2}, \S \right\}. 
    \end{equation}
        
        Now, we consider its swapped PDO $\R'_{BA}$, i.e., $\R'=\S \R \S^\dagger$. We also note that $\left\{ \brho_A \ten \frac{\I}{2}, \S \right\} = \left\{ \frac{\I}{2} \ten \brho_A, \S \right\}$:
        \begin{eqnarray}
            \left\{ \frac{\I}{2} \ten \brho_A, \S \right\} &=& \left(\frac{\I}{2} \ten \brho_A\right) \S + \S\left( \frac{\I}{2} \ten \brho_A \right)\\
            &=& \S\S \left(\frac{\I}{2} \ten \brho_A\right) \S + \S\left( \frac{\I}{2} \ten \brho_A \right)\S\S\\
            &=& \S \left(\brho_A \ten\frac{\I}{2}\right) + \left( \brho_A \ten \frac{\I}{2} \right)\S\\
            &=& \left\{ \brho_A \ten \frac{\I}{2}, \S \right\},
        \end{eqnarray} 
        where we used $\S\S = \I$.
        The swapped PDO $\R'$ can be written as $\R' = \left(\mathcal{U} \ten \id \right) \left\{ \frac{\I}{2} \ten \brho_A, \S \right\}$ by switching $A$ and $B$ as the action of the swap operation on $\R$. By using the above observation, we have $\R' = \left(\mathcal{U} \ten \id \right) \left\{ \brho_A \ten \frac{\I}{2}, \S \right\}$. Thus, we show that $\R'$ can be expressed as follows:
        \begin{eqnarray}
            \R' &=& \left(\mathcal{U} \ten \id \right) \left\{ \brho_A \ten \frac{\I}{2}, \S \right\}\\
                &=& \left\{ \left(\mathcal{U} \ten \id \right) \left(\brho_A \ten \frac{\I}{2}\right), \left(\mathcal{U} \ten \id \right) \left(\S\right) \right\}\\
                &=& \left\{ \mathcal{U}\left(\brho_A\right) \ten \frac{\I}{2}, \left(\id \ten \mathcal{U}^\dagger \right) \left(\S\right) \right\}\\
                &=& \left(\id \ten \mathcal{U}^\dagger \right) \left\{\mathcal{U}\left(\brho_A\right) \ten \frac{\I}{2}, \S \right\}.
        \end{eqnarray}
        The third equation holds due to the fact that $(\I \ten \mathbf{M})\ket{\Phi^+} = (\mathbf{M}^T \ten \I)\ket{\Phi^+}$ for any linear operator $\mathbf{M}$, where $\ket{\Phi^+}$ denotes the maximally entangled state and $(\cdot)^{T}$ denotes the transpose map.\\
        
        Hence we conclude that $\R_{AB}$ is reverse temporal with the adjoint channel $\mathcal{U}^\dagger$ as its reverse pseudo-channel and $\mathcal{U}(\brho)$ as its associated initial state.
    \end{proof}
    
\begin{restatable}{exm}{exasym}
    Let $\R_{AB}$ be a PDO constructed from the following temporally distributed quantum systems: a quantum system $A$ prepared in a state $\brho_{A} = \frac{1}{4} \proj{+} + \frac{3}{4}\proj{-}$ evolves to a system $B$ via a quantum channel $\E_{p}: \brho \mapsto p\brho+(1-p)\mathbf{Z}\brho \mathbf{Z}^{\dagger}$ for $p \in (0,1)$. By construction $\R_{AB}$ is forward temporal, however it is not reverse temporal as we have verified by explicit computation of the reverse atemporality (see \cref{fig:fvsr}). \label{exm:asym}
\end{restatable}

\begin{figure}[t]
\centering
\includegraphics[width=0.8\textwidth]{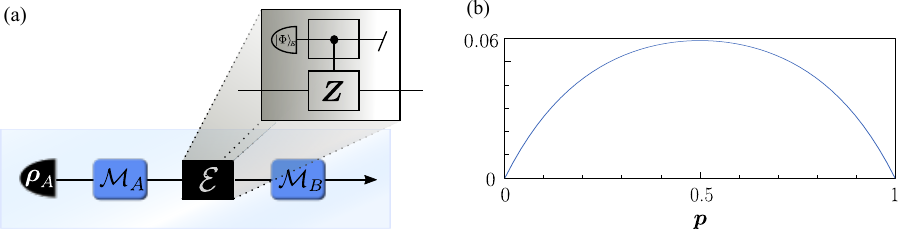}
    \caption{\protect (a) Consider a quantum circuit where a qubit $A$ initially in state $\brho_A = \frac{1}{4} \proj{+} + \frac{3}{4}\proj{-}$ is subject to a quantum channel $\E(\brho):= p\brho + (1-p)\Z\brho\Z\+$, the output of which we label qubit $B$. One can implement this temporal process by a controlled-$Z$ gate with an ancillary system $E$, prepared in a state $\ket{\Phi}_E = \sqrt{p} \ket{0} + \sqrt{1-p}\ket{1}$. Clearly, the correlations between $A$ and $B$ admit a temporal distribution mechanism from $A$ to $B$. (b) However, its reverse atemporality, represented on y-axis, is non-zero and thus has no causal interpretation from $B$ to $A$.}
    \label{fig:fvsr}
\end{figure}

    \cref{fig:swapped} shows the asymmetry of atemporality on randomly sampled density operators. The discrepancy between forward atemporality $\overrightarrow{f}$ and reverse atemporality $\overleftarrow{f}$ is plotted outside of the $x=y$ line. In particular, those on the y-axis indicate the existence of temporal PDOs from the point of view that the system $A$ occurs before the system $B$ but they are not temporal in the reverse point of view. This shows that the forward \at\ does not imply the reverse \at, 
    and vice versa.
\begin{figure}[t]
    \centering
    \includegraphics[scale=0.55]{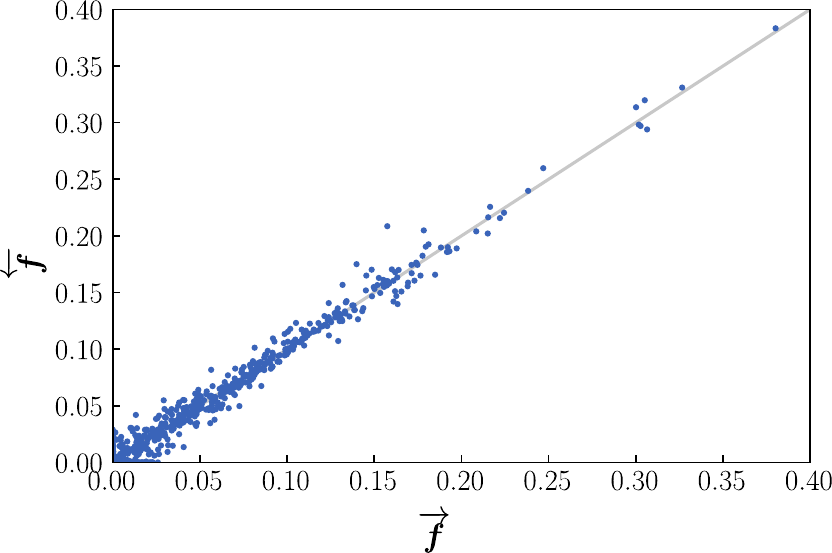}
    \caption{\protect The graph of $\protect\overrightarrow{f}$ against $\protect\overleftarrow{f}$ for 1000 uniformly random PDO $\R \in \mathcal{S}$. The grey line represents the $x=y$ line.}
    \label{fig:swapped}
\end{figure}
    
\begin{restatable}{exm}{exbiasedwer}
    Let $\R_{AB}$ be a PDO given by a biased Werner state $\R = \frac{1}{2} (\brho^\text{Werner}_{q=0.25} + \proj{00})$. 
    Then, it has different values for atemporality $f$ and entanglement negativity $E_\text{neg}$ \cite{vidal02negativity}. More precisely, $f(\R)=0$ and $E_\text{neg}(\R) \approx 0.0087$ (see \cref{fig:bwerner}). \label{mex:atem_ent}
\end{restatable}

\begin{figure}[t]
    \centering
    \includegraphics[scale=0.8]{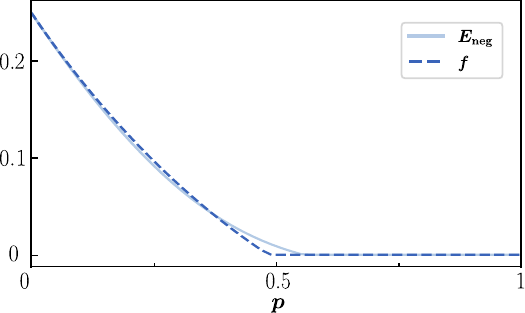}
    \caption{\protect The graphs of entanglement negativity $E_\text{neg}$ and atemporality $f$ of the parameterized family of biased Werner states $\R_{p}^\text{B.W.} \equiv (1-p)\brho^\text{Werner}_{q=0.25} + p\proj{00}$ against $p$. }
    \label{fig:bwerner}
\end{figure}
    A scatter plot (see \cref{fig:mixed}) of atemporality vs. entanglement negativity for 1000 randomly generated density operators suggest that two concepts are heavily correlated but not the same -- with atemporality looking to be a stronger notion of non-classical correlations than entanglement.
    \begin{figure}[ht]
        \includegraphics[scale=0.55]{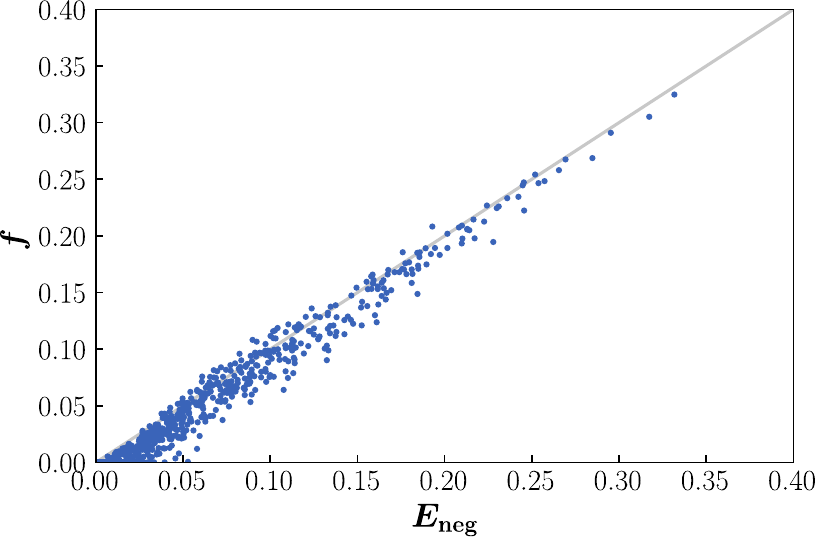}
        \caption{The graph of the entanglement negativity $E_{\text{neg}}$ against the \at\ $f$ for 1000 uniformly random PDO $\R \in \mathcal{S}$. The grey line represents the $x=y$ line.}
        \label{fig:mixed}
    \end{figure}

    \clearpage
\pagebreak
\subsection{Proofs of technical results}\label{app3}

\noindent\textbf{Temporal PDOs.} Imagine that Alice and Bob have qubit systems, $A$ and $B$, respectively, and they had been temporally distributed where measurements on the system $A$ occurs before measurements on the system $B$, namely, the system $A$ is initially prepared in some state $\brho_A$ and the system $B$ is the result of acting some physical time evolution $\E$ on the post-measurement state of the qubit $A$. Each round, Alice and Bob implement Pauli measurements in their choices $a,b$ of basis. In this scenario, the joint probability $\Pr(xy|ab)$ to obtain the measurement outcome $x,y$ is given by
\begin{eqnarray}
    \Pr(xy|ab) = \tr\bigl[ \E\left(\bPi_{x|a}\brho_A\bPi_{x|a}\right) \bPi_{y|b} \bigr],
    \label{eq:app_temporal_pr}
\end{eqnarray}
where $\bPi_{x|a},\bPi_{y|b}$ denotes the projection operator to $x,y$ of the Pauli operator $\bsigma_a,\bsigma_b$, respectively. 

Then it is led to the known results in Ref.~\cite{zhao18pdogeometry,horsman17pdoandothers}, that is, it was shown that the corresponding PDO $\R_{AB}$ can be written by its associated initial state $\brho_A$ of the system $A$ and its associated physical time evolution $\E:A \to B$ as 
	\begin{equation}
        \R_{AB} = \left(\id \ten \E\right) \K, \qquad \mathrm{where} \,\, \K = \left\{ \brho_A \ten \frac{\I}{2}, \swap \right\}.  \label{eq:temporal_app}
    \end{equation} 
Here $\swap$ denotes the swap operator, i.e., $\swap \equiv \frac{1}{2}\sum_{k=0}^3 \bsigma_{k} \ten \bsigma_{k}$, and the curly bracket denotes the anti-commutator, i.e., $\{\mathbf{A},\mathbf{B}\} = \mathbf{A}\mathbf{B}+\mathbf{B}\mathbf{A}$. For notational shorthand, we denote $\left\{ \brho_A \ten \frac{\I}{2}, \swap \right\}$ by $\K$ in what follows. Note that $\brho_A$ is a density operator that $\brho_A = \tr_B \R_{AB}$ and $\E$ is a completely positive (CP) and trace-preserving (TP) map. 

Indeed, from Eq.~(\ref{eq:app_temporal_pr}), we have expectation values $\expval{\bsigma_{a},\bsigma_{b}}$: 
\begin{eqnarray}
    \expval{\bsigma_{a},\bsigma_{b}} &\equiv& \sum_{xy} xy  \Pr(xy|ab) \label{eq:exp1}\\
        &=& \sum_{xy} xy  \tr\left[ \E(\bPi_{x|a}\brho_A\bPi_{x|a}) \bPi_{y|b} \right] \label{eq:exp2}\\
        &=& \tr\bigl[ \E\left(\sum_{x} x \bPi_{x|a}\brho_A\bPi_{x|a}\right)  \left(\sum_y y\bPi_{y|b}\right) \bigr]\\
        &=& \tr\left[ \E\left(\frac{1}{2}\left\{\brho_A, \bsigma_a\right\}\right) \bsigma_{b} \right].
\end{eqnarray}
The last equation holds due to the fact that $\bsigma_a = \sum_{x} x \bPi_{x|a}$ and $\bPi_{x|a} = \frac{1}{2}\left( \I + x\bsigma_a \right)$. By using $\bPi_{x|a} = \frac{1}{2}\left( \I + x\bsigma_a \right)$, we indeed have $\sum_{x=-1,+1} x \bPi_{x|a}\brho_A\bPi_{x|a} = \frac{1}{2}(\brho_A \bsigma_a + \bsigma_a\brho_A)$ (recall that the curly bracket here represents the anti-commutator). Then it is led to \cref{eq:temporal_app} by virtue of the cyclic property of the trace and the fact that $\frac{1}{2}\sum_k \tr\left[ \left\{\brho_A, \bsigma_a\right\}\bsigma_{k} \right]\bsigma_k = \left\{\brho_A, \bsigma_a\right\}$: 
\begin{eqnarray}
    \R_{AB} &\equiv& \frac{1}{4}\sum_{ab}\expval{\bsigma_{a},\bsigma_{b}}\bsigma_a\ten\bsigma_b \\
        &=& \frac{1}{4}\sum_{ab}\tr\biggl[ \E\left(\frac{1}{2}\left\{\brho_A, \bsigma_a\right\}\right) \bsigma_{b} \biggr]\bsigma_a\ten\bsigma_b \\
        &=& \frac{1}{8}\sum_{ab}\frac{1}{2}\sum_k\tr\bigl[ \left\{\brho_A, \bsigma_a\right\} \bsigma_{k} \bigr]\tr\bigl[\E\left(\bsigma_k\right)\bsigma_b\bigr]\bsigma_a\ten\bsigma_b \\
        &=& \frac{1}{16}\sum_{ab}\sum_k\tr\bigl[ \left\{\brho_A, \bsigma_k\right\} \bsigma_{a} \bigr]\tr\bigl[\E\left(\bsigma_k\right)\bsigma_b\bigr]\bsigma_a\ten\bsigma_b \\
        &=& \frac{1}{16}\sum_k\sum_{ab}\tr\Bigl[ \bigl(\left\{\brho_A, \bsigma_k\right\} \ten \E\left(\bsigma_k\right) \bigr)\bsigma_{a}\ten\bsigma_b\Bigr]\bsigma_a\ten\bsigma_b \\
        &=& \frac{1}{4}\sum_k \left\{\brho_A, \bsigma_k\right\} \ten \E\left(\bsigma_k\right).
\end{eqnarray}
The original derivation of Eq.~(\ref{eq:temporal_app}) was provided in Appendix A of the literature \cite{horsman17pdoandothers}. 

Or equivalently, $\R_{AB}$ can be written as 
\begin{eqnarray}
        \R_{AB} &=&\frac{1}{4}\sum_{ab}\tr\left[ \frac{1}{2}\E\bigl(\left\{\brho_A, \bsigma_a\right\}\bigr) \bsigma_{b} \right]\bsigma_a\ten\bsigma_b \\
        &=& \frac{1}{4}\sum_{a}\bsigma_a \ten \left(\frac{1}{2}\sum_{b}\tr\Bigl[ \E\bigl(\left\{\brho_A, \bsigma_a\right\}\bigr) \bsigma_{b} \Bigr]\bsigma_b\right) \\
        &=& \frac{1}{4}\sum_{a} \bsigma_a \ten \E\bigl(\left\{\brho_A, \bsigma_a\right\}\bigr)\\
        &=& \left(\frac{1}{4}\sum_{a\ne0} \bsigma_a \ten \E\bigl(\left\{\brho_A, \bsigma_a\right\}\bigr)\right) + \left(\frac{\I}{2} \ten \E\left(\brho_A\right)\right)\\
        &=& \left(\frac{1}{4}\sum_{a\ne0} \bsigma_a \ten \E\bigl( \bsigma_a + \tr[\brho_A\bsigma_a]\I \bigr)\right) + \left(\frac{\I}{2} \ten \E\left(\brho_A\right)\right) \\
        &=& \left(\frac{1}{4}\sum_{a\ne0} \bsigma_a \ten \E\left( \bsigma_a \right) + \frac{1}{4}\sum_{a\ne0} \tr[\brho_A\bsigma_a] \bsigma_a \ten \E\left( \I \right) \right) + \left(\frac{\I}{2} \ten \E\left(\brho_A\right)\right)\\
        &=& \left(\frac{1}{4}\sum_{a} \bsigma_a \ten \E\left( \bsigma_a \right) \right) + \left( \brho_A - \frac{\I}{2} \right) \ten \E\left( \frac{\I}{2} \right) + \frac{\I}{2} \ten \E\left(\brho_A - \frac{\I}{2} \right) \label{eq:49}. 
\end{eqnarray}
In the fifth equation, we used the following:
\begin{eqnarray}
    \forall a \in \left\{1,2,3\right\}, \quad \left\{\brho_A, \bsigma_a\right\} &=& \frac{1}{2} \sum_k \tr\left[\brho_A \bsigma_k\right] \left\{\bsigma_k,\bsigma_a\right\},\\
    &=& \bsigma_a + \frac{1}{2} \sum_{k\ne0} \tr\left[\brho_A \bsigma_k\right] \left\{\bsigma_k,\bsigma_a\right\},\\
    &=& \bsigma_a + \frac{1}{2} \sum_{k\ne0} \tr\left[\brho_A \bsigma_k\right] \left( 2\delta_{k,a}\I \right),\\
    &=& \bsigma_a + \tr[\brho_A\bsigma_a]\I,
\end{eqnarray}
where $\delta_{i,j}$ denotes the Kronecker delta.

In summary, we have three different versions of representing a PDO, which we will make use of: 
\begin{eqnarray}
    \bullet \quad \R_{AB} &=& \frac{1}{4}\sum_{a} \bsigma_a \ten \E\bigl(\left\{\brho_A, \bsigma_a\right\}\bigr), \label{eq:app_ver1}\\
    \bullet \quad \R_{AB} &=& \left(\frac{1}{4}\sum_{a\ne0} \bsigma_a \ten \E\bigl( \bsigma_a + \tr[\brho_A\bsigma_a]\I \bigr)\right) + \frac{\I}{2} \ten \E\left(\brho_A\right), \label{eq:app_ver2}\\
    \bullet \quad \R_{AB} &=& \left(\frac{1}{4}\sum_{a} \bsigma_a \ten \E\left( \bsigma_a \right) \right) + \left( \brho_A - \frac{\I}{2} \right) \ten \E\left( \frac{\I}{2} \right) + \frac{\I}{2} \ten \E\left(\brho_A - \frac{\I}{2} \right). \label{eq:app_ver3}
\end{eqnarray}
We will use different representations specific to respective contexts.\\

\noindent\textbf{Existence of pseudo-channels.} Result 1 in the main text states that $f>0$ is a necessary and sufficient condition for atemporality. By definition, $f > 0$ implies that $\R_{AB}$ is atemporal, that is, no physical channel is compatible with $\R_{AB}$. Meanwhile, $\R_{AB}$ being atemporal implying $f > 0$ requires each $\R_{AB}$ to have at least one compatible forward and reverse pseudo-channel. This is guaranteed by the following lemma: 
  
\begin{restatable}{lem}{exist}
    Any $2$-qubit PDO $\R_{AB}$ has at least one compatible forward (reverse) pseudo-channel. 
    \label{lem:exist}
\end{restatable}

\begin{proof}
Recall that a forward pseudo-channel of a 2-qubit PDO $\R_{AB}$ is defined as a linear trace preserving map $\Lambda$ that satisfies 
	\begin{equation}
        \R_{AB} = \left(\id_A \ten \Lambda\right) \K, \qquad \mathrm{where} \,\, \K = \left\{ \brho_A \ten \frac{\I}{2}, \swap \right\} \,\,\mathrm{and} \,\,\brho_A=\tr_B \R_{AB}. \label{eq:temporal_1}
    \end{equation} 
Observe that if there exists such map satisfying above equation, it follows that  
    \begin{eqnarray}
        \tr_A\bigl[\left(\bsigma_i \ten \I\right)\R_{AB}\bigr] = \Lambda\bigl(\left\{ \brho_A, \bsigma_i \right\}\bigr).
    \end{eqnarray}
Let $\R_{AB}$ be any PDO such that $\R_{AB} = \frac{1}{4} \sum_{i,j} r_{ij} \bsigma_i \ten \bsigma_j$ for $r_{ij} \in \mathbb{R}$ where $r_{00}=1$. Based on the above observation, we can construct a linear (Hermiticity-preserving) map with scalar multiplication as
    \begin{eqnarray}
		\Lambda\bigl(\left\{ \brho_A, \bsigma_i \right\}\bigr) &\equiv& \tr_A\bigl[\left(\bsigma_i \ten \I\right)\R_{AB}\bigr]\\
    &=& \sum_k r_{ik} \bsigma_k, \label{eq:app_construct_pc} 
	\end{eqnarray}
where $\brho_A$ denotes the first marginal of $\R_{AB}$, i.e., $\brho_A \equiv \tr_B \R_{AB}$.\\
 
 We will prove that, (1) $\Lambda$ indeed satisfies that $\R_{AB} = (\id \ten \Lambda) \K$, (2) $\Lambda$ is well-defined, and that (3) $\Lambda$ is a trace-preserving map. \\
    
	(1) Firstly, we show that the map $\Lambda$ satisfies that $\R_{AB} = (\id \ten \Lambda) \K$. As Pauli operators and the identity operator span the space of Hermitian operators and $\R_{AB}$ is a Hermitian operator, we can instead show that $\tr[\bsigma_i \ten \bsigma_j \R_{AB}] = \tr\left[ \bsigma_i \ten \bsigma_j (\id \ten \Lambda) \K\right]$ for any $\bsigma_i,\bsigma_j$. Thus we have
	\begin{eqnarray}
		\tr\left[ \bsigma_i \ten \bsigma_j (\id \ten \Lambda) \K\right] &=& \frac{1}{4} \sum_k \tr\bigl[\bsigma_i \left\{ \brho_A, \bsigma_k \right\}\bigr] \tr\bigl[ \bsigma_j \Lambda\left(\bsigma_k\right)\bigr] \\
		&=& \frac{1}{4} \sum_k \tr\bigl[\left\{ \bsigma_i, \brho_A \right\}\bsigma_k\bigr] \tr\bigl[ \bsigma_j \Lambda\left(\bsigma_k\right)\bigr] \\
		&=& \frac{1}{2} \tr\left[ \bsigma_j \cdot \Lambda\left( \frac{1}{2} \sum_k \tr\bigl[\bsigma_k \left\{ \brho_A, \bsigma_i \right\}\bigr] \bsigma_k\right)\right] \\
		&=& \frac{1}{2} \tr\Bigl[ \bsigma_j \cdot\Lambda\bigl( \left\{ \brho_A, \bsigma_i \right\}\bigr)\Bigr] \\
		&=& \frac{1}{2} \tr\left[ \bsigma_j \sum_k r_{ik} \bsigma_k \right] \\
		&=& \frac{1}{2} \sum_k r_{ik} \tr\left[ \bsigma_j \bsigma_k \right] \\
		&=& \frac{1}{2} \sum_k r_{ik} 2 \delta_{j,k} = r_{ij} = \tr[\bsigma_i \ten \bsigma_j \R].
	\end{eqnarray}
	Recall that $\K = \frac{1}{4}\sum_k \{\brho_A,\bsigma_k\}\ten\bsigma_k$. The fourth equation holds because Pauli operators and the identity operator form the complete basis of the space of all linear operators and thus $ \left\{ \brho_A, \bsigma_i \right\} = \frac{1}{2} \sum_k \tr[\bsigma_k \left\{ \brho_A, \bsigma_i \right\}] \bsigma_k$. By definition of the map $\Lambda$, the fifth equation follows. \\
	
	(2) Secondly, we will show that the map $\Lambda$ is a well-defined map on the space spanned by $\left\{ \brho_A, \bsigma_{i} \right\}, ~ \forall i$. To verify whether the map $\Lambda$ is well defined on such space, we need to prove that if $\left\{ \brho_A, \bsigma_{i_1} \right\} = \left\{ \brho_A, \bsigma_{i_2} \right\}$, then $\Lambda\left(\left\{ \brho_A, \bsigma_{i_1} \right\}\right) = \Lambda\left(\left\{ \brho_A, \bsigma_{i_2} \right\}\right)$.
Assume that $\left\{ \brho_A, \bsigma_{i_1} \right\} = \left\{ \brho_A, \bsigma_{i_2} \right\}$. \\
	
	(2a) If $i_1 \ne 0 $ and $i_2 \ne 0$,
	
	\begin{eqnarray}
		\left\{ \brho_A, \bsigma_{i_1} \right\} = \left\{ \brho_A, \bsigma_{i_2} \right\} &\Rightarrow& \sum_j \frac{r_{j,0}}{2} \left\{ \bsigma_{j}, \bsigma_{i_1} \right\} = \sum_j \frac{r_{j,0}}{2} \left\{ \bsigma_{j}, \bsigma_{i_2} \right\} \\
		&\Rightarrow& \bsigma_{i_1} + r_{i_1,0}\I = \bsigma_{i_2} + r_{i_2,0}\I \label{eq:LH}
	\end{eqnarray}
	Note that $\brho_A = \sum_j \dfrac{r_{j,0}}{2} \bsigma_j$ and $r_{00}=1$. Multiplying both hand sides of Eq. (\ref{eq:LH}) by $\bsigma_{i_1}$ and taking a trace of the entire term, we have $i_1 = i_2$ because $2 = \tr[\bsigma_{i_1}\bsigma_{i_2}]$. It follows that 
	\begin{eqnarray}
		\Lambda\left(\left\{ \brho_A, \bsigma_{i_1} \right\}\right) &=& \Lambda\left(\left\{ \brho_A, \bsigma_{i_2} \right\}\right).
	\end{eqnarray}
	
	(2b) Otherwise, without loss of generality, let $i_1 = 0$.\\
	
	It is trivial that if $i_2=i_1=0$, $\Lambda\left(\left\{ \brho_A, \bsigma_{i_1} \right\}\right) = \Lambda\left(\left\{ \brho_A, \bsigma_{i_2} \right\}\right)$. Now assume that $i_2 \ne 0$, and see if the equation $\Lambda\left(\left\{ \brho_A, \bsigma_{i_1} \right\}\right) = \Lambda\left(\left\{ \brho_A, \bsigma_{i_2} \right\}\right)$ still holds. The initial assumption of $\left\{ \brho_A, \bsigma_{i_1} \right\} = \left\{ \brho_A, \bsigma_{i_2} \right\}$ implies that $\brho_A = \frac{1}{2}(\I+\bsigma_{i_2})$. Let $\mathbf{\Pi}_{x|a}$ denote the projector of a Pauli operator $\bsigma_{a}$ associated with its eigenvalue $x$, that is, $\bsigma_{a} = \sum_{x=\pm1} x \mathbf{\Pi}_{x|a}$, and we note that $\mathbf{\Pi}_{x|a}=\frac{1}{2}(\I+x\bsigma_{a})$. Thus, $\brho_A$ can be rewritten as $\brho_A = \mathbf{\Pi}_{+1|i_2}$, and furthermore 
    \begin{eqnarray}
		\tr_A[\left(\mathbf{\Pi}_{-1|i_2} \ten \I \right) \R_{AB}] &=& \frac{1}{4} \tr\left[\left\{\brho_A,\bsigma_k\right\}\mathbf{\Pi}_{-1|i_2}\right] \Lambda(\bsigma_k) \\
		&=& \frac{1}{4} \tr\left[\left\{\brho_A,\mathbf{\Pi}_{-1|i_2}\right\}\bsigma_k\right] \Lambda(\bsigma_k) \\
		&=& \zero,
	\end{eqnarray}
    where $\zero$ denotes the zero operator whose elements are all zero. The second equality holds due to the cyclic property of the trace, and the last equality holds because $\mathbf{\Pi}_{+1|i_2}\mathbf{\Pi}_{-1|i_2} =\mathbf{\Pi}_{-1|i_2}\mathbf{\Pi}_{+1|i_2} = \zero$. Together with the fact that $\mathbf{\Pi}_{-1|i_2}=\frac{1}{2}(\I-\bsigma_{a})$, this leads us to $\tr_A \R = \tr_A[\left(\bsigma_{i_2} \ten \I \right) \R]$, i.e., $\sum_k r_{0,k} \bsigma_k = \sum_k r_{i_2,k} \bsigma_k$:  
	\begin{eqnarray}
		\zero = \tr_A[\left(\mathbf{\Pi}_{-1|i_2} \ten \I \right) \R_{AB}] &\Leftrightarrow&  \zero = \tr_A \left[\left( \frac{\I - \bsigma_{i_2}}{2} \ten \I \right)\R_{AB}\right]\\
		  &\Leftrightarrow& \tr_A \R_{AB} = \tr_A \bigl[\left( \bsigma_{i_2} \ten \I \right)\R_{AB}\bigr].
	\end{eqnarray}
	By using $\sum_k r_{0,k} \bsigma_k = \sum_k r_{i_2,k} \bsigma_k $, we have
	\begin{eqnarray}
		\Lambda\bigl(\left\{ \brho_A, \bsigma_{i_1} \right\}\bigr) &=& \Lambda\bigl(\left\{ \brho_A, \bsigma_{0} \right\}\bigr)\\
		&=&  \sum_k r_{0,k} \bsigma_k\\
		&=& \sum_k r_{i_2,k} \bsigma_k \\
		&=& \Lambda\bigl(\left\{ \brho_A, \bsigma_{i_2} \right\}\bigr).
	\end{eqnarray}

    (3) Lastly, we check if the map $\Lambda$ is trace-preserving: the trace of $\left\{ \brho_A, \bsigma_i \right\}$ is given by 
	\begin{eqnarray}
		\tr\left\{ \brho_A, \bsigma_i \right\} &=& 2\tr[\brho_A\bsigma_i] \\
		&=& \sum_{j} r_{j,0} \tr[\bsigma_j \bsigma_i] \\
		&=& 2r_{i,0},
	\end{eqnarray}
	and the trace of $\Lambda\bigl(\left\{ \brho_A, \bsigma_i \right\}\bigr)$ is given by
	\begin{eqnarray}
		\tr\Lambda\bigl(\left\{ \brho_A, \bsigma_i \right\}\bigr) &=& \sum_k r_{ik} \tr\bsigma_k \\
		&=& \sum_{k} r_{i,k} 2\delta_{k,0} \\
		&=& 2r_{i,0}.
	\end{eqnarray}
	
	We have verified that the map $\Lambda$ is trace-preserving by observing that $\tr\left\{ \brho_A, \bsigma_i \right\} = \tr\Lambda(\left\{ \brho_A, \bsigma_i \right\} )$ for $i=1,2,3$. 

    In a very similar way, we can prove the existence of a reverse pseudo-channel, its well-definedness, and its trace-preserving property.\\
    
    Hence, we conclude that for any PDO $\R_{AB}$ there exists a forward pseudo-channel and a reverse pseudo-channel. 
\end{proof}

Although we could prove that the map, defined as Eq.~(\ref{eq:app_construct_pc}), is trace-preserving (TP), what was proven is that it is TP on the domain that $\{\brho_A,\bsigma_i\}_i$ spans (here the curly bracket is the normal bracket that is used to represent sets). Thus, the condition in Eq.~(\ref{eq:app_construct_pc}) is equivalent to trace-preserving requirement on the domain of all linear operators, if $\brho_A$ has full rank. It is because $\{\brho_A,\bsigma_i\}_i$ spans the set of all linear operators, if $\brho_A$ has full rank. Thus, for the cases where $\brho_A$ does not have full rank, we impose the condition of trace-preserving in addition to Eq.~(\ref{eq:app_construct_pc}) for pseudo-channels to satisfy.

\cref{lem:exist} provides a way that any PDOs can be represented as though its subsystems were temporally distributed by means of pseudo-channels.\\

\noindent\textbf{Identifying compatible pseudo-channels.} We introduce a systematic method to identify such compatible pseudo-channels $\overrightarrow{\Lambda}$ (as expressed by their Choi operator) whose existence is guaranteed by \cref{lem:exist}. When $\R_{AB}$ has full rank marginals, $\overrightarrow{\Lambda}$ is unique and the algorithm that returns its normalized Choi operator $\bchi$ is particularly simple (see \cref{alg:pseudo_channel}). From this, the forward atemporality can be directly computed. Here we present the technical details underlying \cref{alg:pseudo_channel} and outline how to extend it to find the set of forward pseudo-channels compatible with a PDO $\R_{AB}$. 


    
 Recall from Eq.~(\ref{eq:49}) that any arbitrary PDO $\R_{AB}$ can be written with its (forward) pseudo-channel $\overrightarrow{\Lambda}$ as
	\begin{eqnarray}
            \R_{AB} = \frac{1}{4} \sum_{k} \bsigma_k \ten \overrightarrow{\Lambda}(\bsigma_k) + \left( \brho_A-\frac{\I}{2} \right) \ten \overrightarrow{\Lambda}\left(\frac{\I}{2}\right) + \frac{\I}{2} \ten \overrightarrow{\Lambda}\left(\brho_A - \frac{\I}{2}\right), \label{eq:map}
	\end{eqnarray}
	where $\brho_A$ denotes the first marginal of $\R_{AB}$, i.e., $\brho_A \equiv \tr_B \R_{AB}$. We achieve the last equation by using $\brho_A = \frac{1}{2} \sum_j \tr[\brho_A \bsigma_j] \bsigma_j$ and the anti-commutation relation of Pauli operators, i.e., $\left\{ \bsigma_j, \bsigma_k \right\} = 2\delta_{jk}\I,  ~\forall j\ne 0, ~k \ne 0$. We can then derive four simultaneous equations:
\begin{eqnarray}
    \begin{aligned}
        \overrightarrow{\Lambda}&(\X) + \tr\bigl[\brho_A \X\bigr] \overrightarrow{\Lambda}(\I) = 2\tr_A\bigl[\left(\X \ten \I\right) \R_{AB}\bigr],\\
        \overrightarrow{\Lambda}&(\Y) + \tr\bigl[\brho_A \Y\bigr] \overrightarrow{\Lambda}(\I) = 2\tr_A\bigl[\left(\Y \ten \I\right) \R_{AB}\bigr],\\
        \overrightarrow{\Lambda}&(\, \Z) + \tr\bigl[\brho_A \, \Z \bigr] \overrightarrow{\Lambda}(\I)= 2\tr_A\bigl[\left(\,\Z \ten \I\right)  \R_{AB}\bigr],\\
        \overrightarrow{\Lambda}&(\brho_A) = \tr_A\R_{AB}. \label{eq:deriven_map}
    \end{aligned}
\end{eqnarray}
	As the Pauli operators and the identity operator form a complete basis for the space of linear operators, the left-hand side of the last equation can be expanded in this basis $\overrightarrow\Lambda(\bsigma_i)$'s as $\overrightarrow\Lambda(\brho_A) = \frac{1}{2} \sum_{i} \tr[\brho_A\bsigma_i] \overrightarrow\Lambda\left(\bsigma_{i}\right)$.
    Then one may notice that we have four unknown operators $\overrightarrow\Lambda(\bsigma_i)$ and four simultaneous equations for them. If we can solve the equations to obtain $\overrightarrow\Lambda(\bsigma_i)$, these would in turn completely determine the action of the pseudo-channel $\overrightarrow\Lambda$.

	(1) Indeed, if the first marginal $\brho_A$ of $\R_{AB}$ has full rank, there is a \emph{unique} solution; that is, there is only one (forward) pseudo-channel $\overrightarrow{\Lambda}$ from $A$ to $B$ compatible with $\R_{AB}$. This unique pseudo-channel $\overrightarrow{\Lambda}$ can be constructed efficiently from $\R_{AB}$. Consequently, we can express $\overrightarrow{\Lambda}(\I)$ in terms of the known operators $\R_{AB}$ and $\brho_A$. Then we show that we can identify the terms in the last line of Eq.~(\ref{eq:map}) with an operator $\L\equiv\left(\brho_A-\frac{\I}{2}\right)\ten\tr_A\left[\left(\frac{1}{2}\brho_A^{-1} \ten \I\right) \R_{AB}\right]+\frac{\I}{2}\ten\tr_A\left[\left((\I-\frac{1}{2}\brho_A^{-1}) \ten \I\right) \R_{AB}\right]$. By noting that $(T\ten\id)\bchi_{\overrightarrow{\Lambda}} \equiv \frac{1}{4} \sum_{k} \bsigma_k \ten \overrightarrow{\Lambda}(\bsigma_k)$, we eventually achieve the Choi operator $\bchi_{\overrightarrow{\Lambda}}$ of $\overrightarrow{\Lambda}$:
    \begin{eqnarray}
        \bchi_{\overrightarrow{\Lambda}} = (T \ten \id)(\R_{AB} - \L),
    \end{eqnarray}
    where $T$ denotes the transpose map. 

(2) Meanwhile, if the first marginal $\brho_A$ is rank-deficient, there are infinitely many pseudo-channels compatible with $\R_{AB}$. Note that in the qubit case, which is what we currently consider, the rank deficiency of $\brho_A$ implies that it is pure. So we denote $\brho_A$ by $\psi$ for some normalized vector $\psi$, i.e., $\brho_A\equiv\proj{\psi}$. In contrast to the full-rank case, here $\R_{AB}$ does not determine how the pseudo-channel $\overrightarrow{\Lambda}$ acts on a state $\smash{\psi^\perp}$ orthogonal to $\psi$. This leads the four simultaneous equations above to be reduced to three equations, so there is a continuous family of solutions: one for each arbitrary Hermitian operator $\btau=:\overrightarrow{\Lambda}(\proj{\smash{\psi^\perp}})$ of trace one that is assigned to the undetermined output. We assign a Hermitian operator of trace one, because pseudo-channels are Hermiticity-preserving and trace-preserving by definition.; let us denote this solution by $\overrightarrow{\Lambda}_{\btau}$. 
We demonstrate through an example where $\ket{\psi}=\ket{0}$. 
In this case, the third and fourth of the equations in Eq.~(\ref{eq:deriven_map}) turn out to be equivalent:

\begin{eqnarray*}
    \begin{aligned}
        &\overrightarrow{\Lambda}\left(\proj{0}\right) = \tr_A\R_{AB}\\
        \Leftrightarrow&
        \overrightarrow{\Lambda}\left(\Z\right) + \overrightarrow{\Lambda}(\I) = 2\tr_A\R_{AB}\\
        \Leftrightarrow&
        \overrightarrow{\Lambda}\left(\Z\right) + \tr\bigl[\proj{0}\Z\bigr]\overrightarrow{\Lambda}(\I) = 2\tr_A\bigl[\left(\Z \ten \I\right)\R_{AB}\bigr].
    \end{aligned}
\end{eqnarray*}
We have only three independent conditions but four unknown operators. Now, we assign an arbitrary Hermitian operator $\btau$ to $\overrightarrow{\Lambda}(\proj{1})$, and then the resulting pseudo-channel $\overrightarrow{\Lambda}_{\btau}$ satisfies
\begin{eqnarray*}
    \begin{aligned}
        \overrightarrow{\Lambda}_{\btau}&\left(\X\right) + \tr\bigl[\proj{0} \X\bigr] \overrightarrow{\Lambda}_{\btau}(\I) := 2\tr_A\bigl[\left(\X \ten \I\right) \R_{AB}\bigr], \\
        \overrightarrow{\Lambda}_{\btau}&\left(\Y\right) + \tr\bigl[\proj{0} \Y\bigr] \overrightarrow{\Lambda}_{\btau}(\I) := 2\tr_A\bigl[\left(\Y \ten \I\right) \R_{AB}\bigr] ,\\
        \overrightarrow{\Lambda}_{\btau}&\left(\,\Z\right) + \tr\bigl[\proj{0} \,\Z\bigr] \overrightarrow{\Lambda}_{\btau}(\I):= 2\tr_A\bigl[\left(\,\Z \ten \I\right) \R_{AB}\bigr],\\
    \overrightarrow{\Lambda}_{\btau}&\left(\proj{0}\right) := \tr_{A}\R_{AB},\\
    \overrightarrow{\Lambda}_{\btau}&\left(\proj{1}\right) := \btau.
    \end{aligned}
\end{eqnarray*}
These now completely determine the action of the map $\overrightarrow{\Lambda}_{\btau}$. We can determine $\overrightarrow{\Lambda}(\I)$ by linearity, as the sum of $\overrightarrow{\Lambda}(\proj{0}) = \tr_A \R_{AB}$ (which follows by the basic definition of a PDO) and $\overrightarrow{\Lambda}(\proj{1}) = \btau$ (by our assignment), and determine $\overrightarrow{\Lambda}(\X), \overrightarrow{\Lambda}(\Y),$ and $\overrightarrow{\Lambda}(\Z)$ accordingly. Essentially the same procedure works for an arbitrary $\psi$.

Now we seek for a PDO with a full-ranked marginal whose unique pseudo-channel is given specifically by $\overrightarrow{\Lambda}_{\btau}$. Indeed, we can show that an operator $\R_{AB}'$, defined as
\begin{eqnarray}
    \begin{aligned}
        \R_{AB}'&:= \R_{AB} + \left( \frac{\I}{2} - \brho_A\right) \ten \frac{1}{2}\biggl(\btau + \tr_A\R_{AB}\biggr) + \frac{\I}{2} \ten \frac{1}{2}\biggl(\btau - \tr_A\R_{AB}\biggr),
    \end{aligned}
\end{eqnarray}
is a PDO whose marginal is $\frac{\I}{2}$ and pseudo-channel is $\Lambda_{\tau}$. This enables us to use the same procedure that we used for full-ranked case, so we can readily obtain the desired pseudo-channel $\overrightarrow{\Lambda}_{\btau}$. In fact, the Choi operator $\bchi_{\overrightarrow{\Lambda}_{\btau}}$ can be achieved by simply taking a partial transpose of $\R_{AB}'$.

\cref{thm:recovermap} summarizes the procedure of recovering pseudo-channels from a given PDO, and provides a closed form expression of the uniquely determined pseudo-channel for the case of PDOs with a full-ranked marginal.

\begin{restatable}{thm}{recovermap}
	For any PDO $\R_{AB}$, there exists a systematic pseudo-channel recovering method that returns the set $\overrightarrow{\mathcal{C}}(\R_{AB})$ of the corresponding Choi operators of forward pseudo-channels compatible with $\R_{AB}$. In particular, if the associated initial state, namely, $\brho_A = \tr_B \R_{AB}$ has full rank, the unique element $\bchi$ in $\overrightarrow{\mathcal{C}}(\R_{AB})$ is given by the following closed-form expression,
    \begin{equation}
		\bchi = \left(T\ten \id\right) \left(\R_{AB} - \L_{AB}\right),
	\end{equation}
	where $T$ denotes the transpose map and
	\begin{eqnarray}
		\begin{aligned}
			\L_{AB} & \equiv \left(\brho_A-\frac{\I}{2}\right)\ten\tr_A\left[\Bigl(\frac{1}{2}\brho_A^{-1} \ten \I\Bigr) \R_{AB}\right] +\frac{\I}{2}\ten\tr_A\left[\Bigl( \bigl(\I-\frac{1}{2}\brho_A^{-1}\bigr) \ten \I\Bigr) \R_{AB}\right].
		\end{aligned}
	\end{eqnarray}
	Else if $\brho_A$ is rank deficient, $\overrightarrow{\mathcal{C}}(\R_{AB})$ is given by
\begin{eqnarray}
    \overrightarrow{\mathcal{C}}(\R_{AB}) &= \Big\{ (T\ten\id)\R_{\btau}: \tr\btau = 1, \btau \text{ is Hermitian} \Big\},
\end{eqnarray}
 where $\R_{\btau} := \R_{AB} + \left( \frac{\I}{2} - \brho_A\right) \ten \frac{1}{2}\left(\btau + \tr_A\R_{AB}\right) + \frac{\I}{2} \ten \frac{1}{2}\left(\btau -\tr_A\R_{AB}\right)$.
	\label{thm:recovermap}
\end{restatable}
	
\begin{proof}
    Let $\R_{AB}$ be a PDO and $\brho_A$ its first marginal, i.e., $\brho_A \equiv \tr_B \R_{AB}$. As it is clear that we are only concerned a forward pseudo-channel in this proof, we omit the symbol $\rightarrow$ from the pseudo-channels for brevity.\\
    
    (1) Suppose that $\brho_A$ has full rank. Then $\brho_A$ is invertible, so there exists its inverse $\brho_A^{-1}$. By \cref{lem:exist}, $\R_{AB}$ has a pseudo-channel $\Lambda$ that satisfies $\R_{AB} = (\id \ten \Lambda) \K = \frac{1}{4}\sum_k \left\{\brho_A, \bsigma_k \right\} \ten \Lambda(\bsigma_k)$. 
        In turn, we show that the pseudo-channel $\Lambda$ is uniquely determined and its Choi operator $\bchi_{\Lambda}$ is given by $\bchi_{\Lambda} = \left(T\ten \id\right) \left(\R_{AB} - \L_{AB}\right)$, where  
        \begin{eqnarray}
            \begin{aligned}
                \L_{AB} \equiv \left(\brho_A-\frac{\I}{2}\right)\ten\tr_A\left[\left(\frac{1}{2}\brho_A^{-1} \ten \I\right) \R_{AB}\right]+\frac{\I}{2}\ten\tr_A\left[\left((\I-\frac{1}{2}\brho_A^{-1}) \ten \I\right) \R_{AB}\right].
            \end{aligned}
	   \end{eqnarray}
    
        Firstly, we show that 
		\begin{equation}
			\tr_A\Bigl[\left(\brho_A^{-1} \ten \I\right) \R_{AB}\Bigr] = \Lambda(\I) \label{eq:EI}
		\end{equation}
		and 
		\begin{equation}
			\tr_A \R_{AB} = \Lambda(\brho_A), \label{eq:Erho}
		\end{equation}
		then it follows that
		\begin{equation}
			\L_{AB} = \left(\brho_A-\frac{\I}{2}\right)\ten\Lambda\left(\frac{\I}{2}\right) +\frac{\I}{2}\ten\Lambda\left(\brho_A-\frac{\I}{2}\right).
		\end{equation}
		Indeed, Eq.~(\ref{eq:EI}) and Eq.~(\ref{eq:Erho}) hold: \\
		\begin{eqnarray}
			\tr_A\left[\left(\brho_A^{-1} \ten \I\right) \R_{AB}\right] &=& \tr_A\left[\left(\brho_A^{-1} \ten \I\right) \left( \frac{1}{4}\sum_k \{\brho_A, \bsigma_k \} \ten \Lambda(\bsigma_k)\right)\right] \\
			&=& \frac{1}{4} \sum_k \tr\left[\brho_A^{-1} \{\brho_A, \bsigma_k \} )\right]  \Lambda(\bsigma_k)\\
			&=& \frac{1}{4} \sum_k 2\tr[ \bsigma_k ] \Lambda(\bsigma_k) \\
			&=& \frac{1}{2} \sum_k 2\delta_{k,0} \Lambda(\bsigma_k) \\
			&=& \Lambda(\I),
		\end{eqnarray}
		and  
		\begin{eqnarray}
			\tr_A \R_{AB} &=& \tr_A\left[\frac{1}{4}\sum_k \{\brho_A, \bsigma_k \} \ten \Lambda(\bsigma_k)\right] \\
			&=& \sum_k \frac{\tr[\brho_A \bsigma_k]}{2} \Lambda(\bsigma_k) \\
			&=& \Lambda\left( \sum_k \frac{\tr[\brho_A \bsigma_k]}{2} \bsigma_k \right)\\
			&=& \Lambda(\brho_A).
		\end{eqnarray}
		
		Thus, we have $\L_{AB} = \left(\brho_A-\frac{\I}{2}\right)\ten\Lambda\left(\frac{\I}{2}\right) +\frac{\I}{2}\ten\Lambda\left(\brho_A-\frac{\I}{2}\right)$. Now let us recall Eq.~(\ref{eq:app_ver3}) that a PDO $\R_{AB}$ can be written with its pseudo-channel $\Lambda$ and its marginal $\brho_A$ as
        \begin{eqnarray}
            \R_{AB} = \left(\frac{1}{4}\sum_{a} \bsigma_a \ten \Lambda\left( \bsigma_a \right) \right) + \left( \brho_A - \frac{\I}{2} \right) \ten \Lambda\left( \frac{\I}{2} \right) + \frac{\I}{2} \ten \Lambda\left(\brho_A - \frac{\I}{2} \right).
        \end{eqnarray}
        Namely, we have
        \begin{eqnarray}
            \R_{AB} = \bchi_{\Lambda}^{T_{A}} + \L_{AB},
        \end{eqnarray}
        where $(\cdot)^{T_{A}}$ denote the partial transpose on the first system $A$. Recall that the Choi operator $\bchi_{\Lambda}$ is defined as $\bchi_{\Lambda} \equiv\frac{1}{4}\sum_{a} (\bsigma_a)^T \ten \Lambda\left( \bsigma_a \right)$, with the transpose map $(\cdot)^T$.  We then observe that the Choi operator $\bchi_{\Lambda}$ of the pseudo-channel $\Lambda$ is obtained by $\R_{AB}^{T_{A}} - \L_{AB}^{T_{A}}$.\\

    (2) Now suppose that $\brho_A$ does not have full rank, that is, it is pure. The proof is divided into two parts: We show (2a) conditions that a pseudo-channel of $\R_{AB}$ must satisfy, (2b) that there are infinitely many pseudo-channels compatible with $\R_{AB}$, and  (2c) how we can systematically find these pseudo-channels. \\

    (2a) We recall Eq.~(\ref{eq:app_ver2}) that $\R_{AB}$ can be written as
    \begin{eqnarray}
        \R_{AB} = \left(\frac{1}{4}\sum_{a\ne0} \bsigma_a \ten \Lambda\left( \bsigma_a + \tr[\brho_A\bsigma_a]\I \right)\right) + \frac{\I}{2} \ten \Lambda\left(\brho_A\right).
    \end{eqnarray}
From this, we can derive four simultaneous equations that a pseudo-channel $\Lambda$ of $\R_{AB}$ must satisfy:
    \begin{eqnarray}
        \begin{aligned}
        \forall k \in \left\{1,2,3\right\}, \quad \tr_A\bigl[\left(\bsigma_k \ten \I\right)\R_{AB}\bigr] &= \left(\frac{1}{4}\sum_{a\ne0} \tr\left[\bsigma_k\bsigma_a\right]  \Lambda\bigl( \bsigma_a + \tr[\brho_A\bsigma_a]\I \bigr)\right) + \tr\left[\frac{\bsigma_k}{2}\right]  \Lambda\left(\brho_A\right)\\
        &= \frac{1}{4}\sum_{a\ne0} \left(2\delta_{k,a} \right)\Lambda\bigl( \bsigma_a + \tr[\brho_A\bsigma_a]\I \bigr) \\
        &= \frac{1}{2} \Lambda\bigl( \bsigma_k + \tr[\brho_A\bsigma_k]\I \bigr), \label{eq:pc_condi_123}
        \end{aligned}
    \end{eqnarray}
    and 
    \begin{eqnarray}
        \begin{aligned}
        \tr_A\R_{AB} &= \left(\frac{1}{4}\sum_{a\ne0} \tr\left[\bsigma_a\right] \Lambda\bigl( \bsigma_a + \tr[\brho_A\bsigma_a]\I \bigr)\right) + \tr\left[\frac{\I}{2}\right] \Lambda\left(\brho_A\right)\\
        &=\Lambda\left(\brho_A\right). \label{eq:pc_condi_0}
        \end{aligned}
    \end{eqnarray}
    Here we used $\tr[\bsigma_i\bsigma_j]=2\delta_{ij}, ~\forall i,j$ and the traceless property of Pauli operators. By rearranging Eqs.~(\ref{eq:pc_condi_123}) and Eq.~(\ref{eq:pc_condi_0}), we have the following equations:
    \begin{eqnarray}
        \Lambda(\X) + \tr\bigl[\proj{\psi} \X\bigr]\Lambda(\I) &=& 2\tr_A\bigl[(\X \ten \I) \R_{AB}\bigr], \label{eq:pc_con_x} \\ 
        \Lambda(\Y) + \tr\bigl[\proj{\psi} \Y\bigr]\Lambda(\I) &=& 2\tr_A\bigl[(\Y \ten \I) \R_{AB}\bigr], \label{eq:pc_con_y} \\ 
        \Lambda(\, \Z) + \tr\bigl[\proj{\psi} \, \Z\bigr]\Lambda(\I) &=& 2\tr_A\bigl[(\, \Z \ten \I) \R_{AB}\bigr], \label{eq:pc_con_z} \\ 
        \Lambda\left(\proj{\psi}\right) &=& \tr_A\R_{AB}, \label{eq:pc_con_i}
    \end{eqnarray}
    where we denote $\brho_A$ by $\proj{\psi}$ for some unit vector $\psi$, because it is pure. \\

    (2b) As the Pauli operators and the identity operator form a complete basis for the space of linear operators, the left hand side of the last equation can be expanded in this basis $\Lambda(\bsigma_i)$'s. Consequently, the last equation become 
    \begin{eqnarray}
        \frac{1}{2} \sum_i \tr\bigl[\proj{\psi}\bsigma_i\bigr]\Lambda(\bsigma_i) = \tr_A\R_{AB}. \label{eq:pc_con_last}
    \end{eqnarray}
    To solve the above simultaneous equations, we plug Eqs.~(\ref{eq:pc_con_x},\ref{eq:pc_con_y},\ref{eq:pc_con_z}) into Eq.~(\ref{eq:pc_con_last}). Then we observe that Eq.~(\ref{eq:pc_con_last}) is redundant, that is, in the rank-deficient cases, we have only three simultaneous equations and four unknown operators $\Lambda(\X), \Lambda(\Y), \Lambda(\Z),$ and $\Lambda(\I)$: 
    \begin{eqnarray}
        & & \frac{1}{2}\Lambda(\I) + \frac{1}{2}\left(\sum_{i\ne0} \tr\bigl[\proj{\psi}\bsigma_i\bigr]\Bigl( 2\tr_A\bigl[\left(\bsigma_i \ten \I\right)\R_{AB}\bigr] - \tr\bigl[\proj{\psi}\bsigma_i\bigr] \Lambda\left(\I\right)\Bigr) \right) = \tr_A\R_{AB},\\ 
        &\Leftrightarrow & \Lambda(\I) + 2 \tr_A\left[\Biggl(\biggl(\sum_{i\ne0} \tr[\proj{\psi}\bsigma_i] \bsigma_i\biggr) \ten \I\Biggr)\R_{AB}\right] - \sum_{i\ne0}\bigl( \tr\left[\proj{\psi}\bsigma_i\right]\bigr)^2 \Lambda(\I) = 2\tr_A\R_{AB}, \\
        &\Leftrightarrow & \Lambda(\I) + 2 \tr_A\Bigl[\bigl( \left(2\proj{\psi} - \I \right) \ten \I\bigr)\R_{AB}\Bigr] - \Lambda(\I) = 2\tr_A\R_{AB},\\
        &\Leftrightarrow & 4\tr_A\bigl[\left( \proj{\psi} \ten \I\right)\R_{AB}\bigr] -2\tr_A\R_{AB} = 2\tr_A\R_{AB},\\
        &\Leftrightarrow & 2\tr_A\R_{AB} = 2\tr_A\R_{AB}\\
        &\Leftrightarrow & \zero = \zero.
    \end{eqnarray}
    Here we used $\sum_{i\ne0} \tr[\proj{\psi}\bsigma_i] \bsigma_i = 2\proj{\psi} - \I$, and due to the purity of $\proj{\psi}$, we used $\sum_{i\ne0}\left( \tr[\proj{\psi}\bsigma_i]]\right)^2=1$. We also used $\tr_A\left[\left( \proj{\psi} \ten \I\right)\R_{AB}\right] = \tr_A[\R_{AB}]$. This can be verified by noting that $\R_{AB} = \frac{1}{4}\sum_{j=0}^3 \left\{ \proj{\psi}, \bsigma_j \right\} \ten \Lambda(\bsigma_j)$ and the cyclic property of the trace:
    \begin{eqnarray}
        \tr_A\bigl[\left( \proj{\psi} \ten \I\right)\R_{AB}\bigr] &=& \tr_A\left[\Biggl( \proj{\psi} \ten \I\Biggr)\Biggl( \frac{1}{4}\sum_{j=0}^3 \left\{ \proj{\psi}, \bsigma_j \right\} \ten \Lambda(\bsigma_j) \Biggr)\right], \\
        &=& \frac{1}{4} \sum_{j=0}^3 \tr\Bigl[ \proj{\psi} \bigl\{ \proj{\psi}, \bsigma_j \bigr\} \Bigr] \Lambda(\bsigma_j) , \\
        &=& \frac{1}{4} \sum_{j=0}^3 \tr\Bigl[ \proj{\psi} \bigl( \proj{\psi} \bsigma_j + \bsigma_j \proj{\psi}\bigr) \Bigr] \Lambda(\bsigma_j), \\ 
        &=& \frac{1}{4} \sum_{j=0}^3 \tr\Bigl[  \proj{\psi} \bsigma_j + \bsigma_j \proj{\psi} \Bigr] \Lambda(\bsigma_j), \\ 
        &=& \frac{1}{4} \sum_{j=0}^3 \tr\bigl\{ \proj{\psi}, \bsigma_j \bigr\} \Lambda(\bsigma_j), \\ 
        &=& \tr_A\left[\frac{1}{4}\sum_{j=0}^3 \bigl\{ \proj{\psi}, \bsigma_j \bigr\}
 \ten \Lambda(\bsigma_j) \right] = \tr_A \R_{AB}.
    \end{eqnarray}
    In rank-deficient cases, we do not know the action of $\Lambda$ on the perpendicular state $\smash{\psi^{\perp}}$ to $\psi$, that is, we do not know  $\Lambda(\proj{\smash{\psi^{\perp}}})$. Let us denote $\Lambda(\proj{\smash{\psi^{\perp}}})$ by some operator $\btau$. Because of this unknown operator $\btau$, we cannot determine $\Lambda(\X), \Lambda(\Y), \Lambda(\Z),$ and $\Lambda(\I)$, and consequently cannot determine $\Lambda$. 
    As pseudo-channels are defined to be Hermiticity-preserving and trace-preserving, we can only consider $\btau$ which are Hermitian operators of trace one.  
    
    By assigning $\btau$ to $\Lambda(\proj{\smash{\psi^{\perp}}})$, we can specify one pseudo-channel among all pseudo-channels, and we denote this particular pseudo-channel by $\Lambda_{\btau}$. That is, by construction, $\Lambda_{\btau}$ is a linear trace-preserving map that satisfies
        \begin{eqnarray}
        \Lambda_{\btau}(\X) + \tr\bigl[\proj{\psi} \X\bigr] \Lambda_{\btau}(\I) &=& 2\tr_A\bigl[\left(\X \ten \I\right) \R_{AB}\bigr] , \label{eq:app_L_tau_X}\\ 
        \Lambda_{\btau}(\Y) + \tr\bigl[\proj{\psi} \Y\bigr]\Lambda_{\btau}(\I) &=& 2\tr_A\bigl[\left(\Y \ten \I\right) \R_{AB}\bigr], \\ 
        \Lambda_{\btau}(\, \Z) + \tr\bigl[\proj{\psi} \, \Z\bigr]\Lambda_{\btau}(\I) &=& 2\tr_A\bigl[(\, \Z \ten \I)\, \R_{AB}\bigr], \\ 
        \Lambda_{\btau}(\proj{\psi}) &=& \tr_A\R_{AB},\\
        \Lambda_{\btau}(\I-\proj{\psi}) &=& \btau. \label{eq:app_L_tau_last}
    \end{eqnarray}
    Note that $\I-\proj{\psi} = \proj{\smash{\psi^{\perp}}}$.\\

    (2c) In turn, we provide a systematic way to find the Choi operator $\bchi_{\Lambda_{\btau}}$ of $\Lambda_{\btau}$. We fist find a PDO $\R_{AB}'$ whose associated initial state $\brho'_A\equiv \tr_A \R_{AB}'$ is $\frac{\I}{2}$ and (unique) pseudo-channel $\Lambda'$ is $\Lambda_{\btau}$, and then we achieve $\bchi_{\Lambda_{\btau}}$ by taking a partial transpose of $\R_{AB}'$. We first recall Eq.~(\ref{eq:app_ver2}) that $\R_{AB}$ can be written in terms of its marginal $\brho_A$ and its pseudo-channel $\Lambda_{\btau}$ as
    \begin{eqnarray}
        \R_{AB} = \left(\frac{1}{4}\sum_{a\ne0} \bsigma_a \ten \Lambda_{\btau}\bigl( \bsigma_a + \tr\left[\brho_A\bsigma_a\right]\I \bigr)\right) + \frac{\I}{2} \ten \Lambda_{\btau}\left(\brho_A\right).
    \end{eqnarray}
    Also note the following equations that $\Lambda_{\btau}$ must satisfy. These are the equations we get from Eqs.~(\ref{eq:app_L_tau_X}-\ref{eq:app_L_tau_last}):
    \begin{eqnarray}
        \forall a \in \left\{1,2,3\right\}, \quad \Lambda_{\btau}\bigl(\bsigma_a + \tr\left[\brho_A \bsigma_a\right]\I \bigr) &=& 2\tr_A\bigl[\left(\bsigma_a \ten \I\right) \R_{AB}\bigr], \label{eq:app_lambda_1} \\ 
        \Lambda_{\btau}(\brho_A) &=& \tr_A\R_{AB}, \label{eq:app_lambda_2}\\
        \Lambda_{\btau}(\I-\brho_A) &=& \btau. \label{eq:app_lambda_3} 
    \end{eqnarray}
    That is, we clearly see that  
    \begin{eqnarray}
        \R_{AB} = \frac{1}{4} \sum_{a\ne0} \bsigma_a \ten \Bigl( 2\tr_A\bigl[ \left(\bsigma_a\ten\I\right)\R_{AB} \bigr] \Bigr) + \frac{\I}{2}\ten \tr_A\R_{AB}. \label{eq:app_R}
    \end{eqnarray}
    
    In the same manner, $\R_{AB}'$ can also be written as
    \begin{eqnarray}
        \R_{AB}' &=& \left(\frac{1}{4}\sum_{a\ne0} \bsigma_a \ten \Lambda'\bigl( \bsigma_a + \tr\left[\brho_A'\bsigma_a\right]\I \bigr)\right) + \frac{\I}{2} \ten \Lambda'\left(\brho_A'\right), \\
        &=& \left(\frac{1}{4}\sum_{a\ne0} \bsigma_a \ten \Lambda_{\btau}\left( \bsigma_a + \tr\left[\left(\frac{\I}{2}\right)\bsigma_a\right]\I \right)\right) + \frac{\I}{2} \ten \Lambda_{\btau}\left(\frac{\I}{2}\right), \\
        &=& \frac{1}{4}\sum_{a=0}^3 \bsigma_a \ten \Lambda_{\btau}(\bsigma_a). 
    \end{eqnarray}
    Since the Choi operator $\bchi_{\Lambda_{\btau}}$ of $\Lambda_{\btau}$ is defined as $\frac{1}{4}\sum_{j=0}^3 \left(\bsigma_j\right)^T \ten \Lambda_{\btau}(\bsigma_j)$, we note that 
    \begin{eqnarray}
        \bchi_{\Lambda_{\btau}} = \bigl(\R_{AB}'\bigr)^{T_{A}}, \label{eq:app_choi_pdo}
    \end{eqnarray}
    where $(\cdot)^{T_A}$ denotes the partial transpose on $A$. 
    
It is left to show that $\R_{AB}'$ can be written in terms of $\R_{AB}$, its marginal $\brho_A \equiv \tr_B \R_{AB}$, and $\btau$:
\begin{eqnarray}
    \R_{AB}' &=& \frac{1}{4}\sum_{a\ne0} \bsigma_a \ten \Lambda'\left( \bsigma_a + \tr[\brho_A'\bsigma_a]\I \right) + \frac{\I}{2} \ten \Lambda'\left(\brho_A'\right), \\
        &=& \frac{1}{4}\sum_{a\ne0} \bsigma_a \ten \Lambda_{\btau}\bigl( \bsigma_a + \tr\left[\frac{\I}{2}\cdot\bsigma_a\right]\I \bigr) + \frac{\I}{2} \ten \Lambda_{\btau}\left(\frac{\I}{2}\right), \\
        &=& \frac{1}{4}\sum_{a\ne0} \bsigma_a \ten \Lambda_{\btau}\left( \bsigma_a + \tr[\left(\brho_A - \left(\brho_A - \frac{\I}{2} \right)\right)\cdot\bsigma_a]\I \right) + \frac{\I}{2} \ten \Lambda_{\btau}\left(\left(\brho_A - \left(\brho_A - \frac{\I}{2} \right)\right)\right), \\
        &=& \frac{1}{4}\sum_{a\ne0} \bsigma_a \ten \Lambda_{\btau}\bigl( \bsigma_a + \tr\left[\brho_A\bsigma_a\right]\I \bigr) + \frac{\I}{2} \ten \Lambda_{\btau}\left(\brho_A\right) \\ 
        & & -\frac{1}{4}\sum_{a\ne0} \tr[\left(\brho_A - \frac{\I}{2} \right)\cdot\bsigma_a] \bsigma_a \ten \Lambda_{\btau}\left( \I \right) - \frac{\I}{2} \ten \Lambda_{\btau}\left(\brho_A - \frac{\I}{2}\right), \\
        &=& \frac{1}{4}\sum_{a\ne0} \bsigma_a \ten \Lambda_{\btau}\bigl( \bsigma_a + \tr\left[\brho_A\bsigma_a\right]\I \bigr) + \frac{\I}{2} \ten \Lambda_{\btau}\left(\brho_A\right) \\ 
        & & -\left(\brho_A - \frac{\I}{2} \right)\ten \Lambda_{\btau}\left( \frac{\I}{2} \right) - \frac{\I}{2} \ten \Lambda_{\btau}\left(\brho_A - \frac{\I}{2}\right)\\
        &=& \frac{1}{4}\sum_{a\ne0} \bsigma_a \ten \Bigl( 2\tr_A\left[(\bsigma_a \ten \I) \R_{AB}\right] \Bigr) + \frac{\I}{2} \ten \Bigl(\tr_A\R_{AB}\Bigr) \\ 
        & & +\left( \frac{\I}{2} + \brho_A \right)\ten \frac{1}{2}\biggl( \btau + \tr_A\R_{AB} \biggr) + \frac{\I}{2} \ten \frac{1}{2}\biggl(\btau - \tr_A\R_{AB}\biggr).
\end{eqnarray}
In the last equation, we used Eqs.~(\ref{eq:app_lambda_1},\ref{eq:app_lambda_2},\ref{eq:app_lambda_3}). Lastly, by using Eq.~(\ref{eq:app_R}), we have
\begin{eqnarray}
    \R_{AB}' &=& \R_{AB} + \left( \frac{\I}{2} - \brho_A \right)\ten \frac{1}{2}\biggl( \btau + \tr_A\R_{AB} \biggr) + \frac{\I}{2} \ten \frac{1}{2}\biggl(\btau - \tr_A\R_{AB}\biggr)\\
    &=& \R_{AB} - \frac{1}{2} \brho_A \ten \tr_A\R_{AB} + \frac{1}{2}\bigl( \I - \brho_A \bigr) \ten \btau. \label{eq:app_alt_pdo}
\end{eqnarray}

Hence we conclude that $\bchi_{\Lambda_{\btau}}$ is readily achievable from $\R_{AB}$.
We note that we used a PDO $\R_{AB}'$ with full-ranked marginal which is specifically given by $\frac{\I}{2}$, but any PDOs whose (first) marginal has full rank can be used instead. We merely used this PDO for computational simplicity, so the simple partial transpose results in the desired operator $\bchi_{\Lambda_{\btau}}$.
\end{proof}

In summary, (1) if a given PDO $\R_{AB}$ has a full ranked marginal $\brho_A \equiv \tr_B \R_{AB}$, it has a unique pseudo-channel $\Lambda$ and it can be readily found in the form of its associated Choi operator $\bchi_{\Lambda}$ by using the closed-form expression as
\begin{eqnarray*}
    \bchi_{\Lambda} = (T \ten \id)\left(\R_{AB}-\L_{AB}\right), 
\end{eqnarray*}
with $\L_{AB}\equiv \left(\brho_A-\frac{\I}{2}\right)\ten\tr_A\left[\left(\frac{1}{2}\brho_A^{-1} \ten \I\right) \R_{AB}\right]+\frac{\I}{2}\ten\tr_A\left[\left((\I-\frac{1}{2}\brho_A^{-1}) \ten \I\right) \R_{AB}\right]$. 
(2) If $\R_{AB}$ has a rank-deficient marginal instead, it has infinitely many pseudo-channels $\Lambda_{\btau}$'s ($\btau$ is a Hermitian operator with unit trace of free choice). Each $\Lambda_{\btau}$ can be obtained by finding a suitable PDO $\R_{AB}'$ whose marginal has full rank and applying the same closed-form expression above to $\R_{AB}'$.\\
 
\noindent\textbf{Sufficient condition for atemporality and entanglement negativity to coincide.} Here we present a theorem underlying Result 3 in the main text. 

\begin{restatable}{thm}{atempent}
Let $\R_{AB}$ be a density operator.
Its \at\ is equal to its entanglement negativity, $f(\R_{AB}) = E_\text{neg}(\R_{AB})$, whenever $\R_{AB}$ is pure or marginals of $\R_{AB}$ are given by the maximally mixed state. \label{thm:atempent}
\end{restatable}
	
	\begin{proof}
		Recall the definitions of $E_\text{neg}$ and $f$, 
		\begin{eqnarray*}
			E_\text{neg}(\R_{AB}) &\equiv& \frac{\norm{\R^{T_A}_{AB}}_{\tr} - 1}{2},
		\end{eqnarray*}
		\begin{eqnarray*}
			f(\R_{AB}) \equiv \min \left( \overrightarrow{f}(\R_{AB}), \overleftarrow{f}(\R_{AB}) \right),
		\end{eqnarray*} 
		where the forward, reverse atemporalities $\overrightarrow{f}, \overleftarrow{f}$ are defined as 
		\begin{eqnarray}
            \overrightarrow{f}(\R_{AB})&\equiv& \min_{\bchi_{\overrightarrow{\Lambda}} \in \overrightarrow{\mathcal{C}}(\R_{AB})} \frac{\norm{\bchi_{\overrightarrow{\Lambda}}}_{\tr} -1}{2}, \\
            \overleftarrow{f}(\R_{AB})&\equiv& \min_{\bchi_{\overleftarrow{\Lambda}} \in \overleftarrow{\mathcal{C}}(\R_{AB})} \frac{\norm{\bchi_{\overleftarrow{\Lambda}}}_{\tr} -1}{2},
		\end{eqnarray}
		respectively. Recall that $\bchi_{\Lambda}$ denotes the corresponding Choi operator to a linear map $\Lambda$, and $\overrightarrow{\mathcal{C}},\overleftarrow{\mathcal{C}}$ denote the sets of Choi operators of all forward, reverse pseudo-channels, respectively. Also note that the set $\overleftarrow{\mathcal{C}}(\R_{AB})$ is equal to the set $\overrightarrow{\mathcal{C}}(\R'_{BA})$, where $\R'_{BA}$ represents the swapped PDO of $\R_{AB}$ , i.e., $\R'_{BA}=\S (\R_{AB}) \S^{\dagger}$.\\
		
		Let $\R_{AB}$ be a density operator such that its marginals have full rank.
		
		(1) Firstly, we prove that if marginals of $\R_{AB}$ are given by the maximally mixed state, $f(\R_{AB}) = E_\text{neg}(\R_{AB})$. Assume that $\tr_A \R_{AB} = \tr_B \R_{AB} = \frac{\I}{2}$. Since the marginals have full rank, both forward pseudo-channel $\overrightarrow{\Lambda}$ and reverse one $\overleftarrow{\Lambda}$ are unique, and by \cref{thm:recovermap} the Choi operators $\bchi_{\overrightarrow{\Lambda}}, \bchi_{\overleftarrow{\Lambda}}$ of them are obtained by
        \begin{eqnarray}
            \bchi_{\overrightarrow{\Lambda}} = (T_A \ten \id_B)(\R_{AB} - \L_{AB}),\\
            \bchi_{\overleftarrow{\Lambda}} = (T_B \ten \id_A)(\R'_{BA} - \L'_{AB}).
        \end{eqnarray}
        Observe that $\L$ and $\L'$ both vanish because both marginals of $\R_{AB}$ are maximally mixed states, so we have $\bchi_{\overrightarrow{\Lambda}} = \R_{AB}^{T_A}$ and $\bchi_{\overleftarrow{\Lambda}} = \R'^{T_B}_{BA}$, resulting in
		\begin{eqnarray}
			\overrightarrow{f}(\R)&=\dfrac{\norm{\R_{AB}^{T_A}}_{\tr} - 1}{2}, \\ 
			\overleftarrow{f}(\R)&=\dfrac{\norm{\R'^{T_B}_{BA}}_{\tr} - 1}{2}.
		\end{eqnarray}
		As $\R'^{T_B}_{BA} = \swap(\R_{AB}^{T_A})\swap^\dagger$ and the swap operation does not change eigenvalues of $\R_{AB}^{T_A}$, $\overrightarrow{f}(\R_{AB}) = \overleftarrow{f}(\R_{AB})$. So it follows that $f(\R_{AB}) = \frac{\norm{\R^{T_A}_{AB}}_{\tr} -1}{2}$. Hence, we conclude $f(\R_{AB}) = E_\text{neg}(\R_{AB})$ for any PDOs whose marginals are given by the maximally mixed state. \\
		
		(2) Secondly, we prove that if $\R$ is pure, $f(\R_{AB}) = E_\text{neg}(\R_{AB})$. 
        In the case where a marginal of $\R$ does not have full rank, it is trivial that $\R$ is a product state, consequently, it has no entanglement and it is temporal, i.e., $f(\R_{AB}) = E_\text{neg}(\R_{AB})=0$. Now, suppose that marginals have full rank.
        We will show that $\overrightarrow{f}(\R_{AB}) = E_\text{neg}(\R_{AB})$, and together with the fact that $\overrightarrow{f}(\R_{AB}) = \overleftarrow{f}(\R_{AB})$, we will complete the proof.
		
		As $\R_{AB}$ is pure, we denote the PDO by $\proj{\psi}$ such that $\ket{\psi} = \sum_{i} \psi_i \ket{e_i} \ten \ket{f_i}$ for some orthonormal bases $\{e_i\}_i,\{f_i\}_i$, where $\sum_{i} \abs{\psi_i}^2 = 1$. To show $\overrightarrow{f}(\R_{AB}) = E_\text{neg}(\R_{AB})$, we first compute $E_\text{neg}(\R_{AB})$. $E_\text{neg}(\R_{AB})$ is equal to the absolute value of the sum over all negative eigenvalues of $\R_{AB}^{T_A}$. So we need to solve the characteristic equation of $\R_{AB}^{T_A}$ to obtain the eigenvalues of $\R_{AB}^{T_A}$:
		\begin{eqnarray}
			0 &=& \det(\R_{AB}^{T_A} - \lambda\I) \\
			  &=& \det(\U\left(\R_{AB}^{T_A} - \lambda\I\right)\U\+)\\
			  &=& \det(\U\R_{AB}^{T_A}\U\+ - \lambda\I).
		\end{eqnarray}
		Here $\U$ is a unitary such that $\U \equiv \sum_{i,j} \ketbra{i}{e_i^*} \ten \ketbra{j}{f_j}$ where $\{\ket{i}\}_i$,$\{\ket{j}\}_j$ are the canonical bases and $\bra{e_i^*}$ is the transpose of $\ket{e_i}$. The second equation holds due to the multiplicativity of the determinant and $\det\U\det\U\+ = 1$. Now we write $\U\R_{AB}^{T_A}\U\+$ in matrix form in the canonical bases $\{\ket{i}\}_i,\{\ket{j}\}_j$ as follows,
		\begin{eqnarray}
			\U\R_{AB}^{T_A}\U\+ &=& \U\left(\sum_{ij}\psi_i\psi_j^* \ketbra{e_j^*}{e_i^*} \ten \ketbra{f_i}{f_j}\right)\U\+\\
			&=& \sum_{ij}\psi_i\psi_j^* \ketbra{j}{i} \ten \ketbra{i}{j}\\
			&=& \begin{bmatrix} 
				\abs{\psi_0}^2 & 0 & 0 & 0 \\
				0 & 0 & \psi_0^*\psi_1 & 0 \\
				0 & \psi_0\psi_1^* & 0 & 0 \\
				0 & 0 & 0 & \abs{\psi_1}^2 
			\end{bmatrix}.
		\end{eqnarray}
		After some computation, we can readily find the eigenvalues of $\R_{AB}^{T_A}$:
		\begin{eqnarray}
			0  &=& \det\left(\U\left(\R_{AB}^{T_A} - \lambda\I\right)\U\+\right)\\
			&=& \det
			\begin{bmatrix} 
				\abs{\psi_0}^2-\lambda & 0 & 0 & 0 \\
				0 & -\lambda & \psi_0^*\psi_1 & 0 \\
				0 & \psi_0\psi_1^* & -\lambda & 0 \\
				0 & 0 & 0 & \abs{\psi_1}^2-\lambda 
			\end{bmatrix}\\
			&=& \left(\lambda - \abs{\psi_0}^2\right)\left(\lambda - \abs{\psi_1}^2\right)\left( \lambda^2 - \abs{\psi_0}^2\abs{\psi_1}^2\right).
		\end{eqnarray}
		Then it has four eigenvalues, $\abs{\psi_0}^2,\abs{\psi_1}^2, \pm \abs{\psi_0}\abs{\psi_1}$. Thus, we have $E_{\text{neg}}(\R_{AB}) = \abs{\psi_0}\abs{\psi_1}$.\\
		
		In turn, we compute $\overrightarrow{f}(\R_{AB})$. By assumption, the first marginal $\brho_A$ of $\R_{AB}$ has full rank, thus $\R_{AB}$ has a unique pseudo-channel $\Lambda$ that is given by $\bchi_{\Lambda} = \R_{AB}^{T_A} - \L_{AB}^{T_A}$. Now let us calculate the eigenvalues of $\bchi_{\Lambda}$: 
		\begin{eqnarray}
			0 &=& \det(\bchi_{\Lambda} - \lambda\I) \\
			  &=& \det(\R_{AB}^{T_A} - \L_{AB}^{T_A} - \lambda\I) \\
			  &=& \det(\U\left(\R_{AB}^{T_A} - \L_{AB}^{T_A} - \lambda\I\right)\U\+).
		\end{eqnarray}
		Let us write $\U(\L_{AB}^{T_A})\U\+$ in matrix form in the canonical bases $\{\ket{i}\}_i, \{\ket{j}\}_j$,
		\begin{eqnarray}
			\U\L_{AB}^{T_A}\U\+ &=& \U\left(
			\left(\brho_A^{T}-\frac{\I}{2}\right)\ten\tr_A\left[\left(\frac{1}{2}\brho_A^{-1} \ten \I\right) \R_{AB}\right]  +\frac{\I}{2}\ten\tr_A\left[\left(\left(\I-\frac{1}{2}\brho_A^{-1}\right) \ten \I\right) \R_{AB}\right]\right)\U\+\\
			&=& \U\left(
			\left(\brho_A^{T}-\frac{\I}{2}\right)\ten\tr_A\left[\left(\frac{1}{2}\brho_A^{-1} \ten \I\right) \R_{AB}\right]  + \frac{\I}{2}\ten
			\left(\btau_B - \tr_A\left[\left(\frac{1}{2}\brho_A^{-1} \ten \I\right) \R_{AB}\right]\right)\right)\U\+\\
			&=& \U\left(
			\left(\sum_i \Bigl(\abs{\psi_i}^2 -\frac{1}{2}\Bigr) \proj{e_i^*}\right)\ten \frac{\I}{2}  +\frac{\I}{2}\ten\left(\sum_i \Bigl(\abs{\psi_i}^2 -\frac{1}{2}\Bigr) \proj{f_i}\right) \right)\U\+ \label{eq:trA}\\
			&=& \left(\sum_i \Bigl(\abs{\psi_i}^2 -\frac{1}{2}\Bigr) \proj{i}\right)\ten \frac{\I}{2}  +\frac{\I}{2}\ten\left(\sum_i \Bigl(\abs{\psi_i}^2 -\frac{1}{2}\Bigr) \proj{i}\right)\\		
			&=& \frac{1}{2} \begin{bmatrix} 
				\abs{\psi_0}^2 -\frac{1}{2} & 0 & 0 & 0 \\
				0 & \abs{\psi_0}^2 -\frac{1}{2} & 0 & 0 \\
				0 & 0 & \abs{\psi_1}^2 -\frac{1}{2} & 0 \\
				0 & 0 & 0 & \abs{\psi_1}^2 -\frac{1}{2}
			\end{bmatrix}
			+
			\frac{1}{2} \begin{bmatrix} 
				\abs{\psi_0}^2 -\frac{1}{2} & 0 & 0 & 0 \\
				0 & \abs{\psi_1}^2 -\frac{1}{2} & 0 & 0 \\
				0 & 0 & \abs{\psi_0}^2 -\frac{1}{2} & 0 \\
				0 & 0 & 0 & \abs{\psi_1}^2 -\frac{1}{2}
			\end{bmatrix}\\
			&=& \begin{bmatrix} 
				\abs{\psi_0}^2 -\frac{1}{2} & 0 & 0 & 0 \\
				0 & 0 & 0 & 0 \\
				0 & 0 & 0 & 0 \\
				0 & 0 & 0 & \abs{\psi_1}^2 -\frac{1}{2}
			\end{bmatrix},
		\end{eqnarray}
		where $\brho_A \equiv \tr_B\R_{AB}$ and $\btau_B \equiv \tr_A\R_{AB}$, i.e., $\brho_A = \sum_i \abs{\psi_i}^2 \proj{e_i}$, $\btau_B = \sum_i \abs{\psi_i}^2 \proj{f_i}$. Eq. (\ref{eq:trA}) holds because $\brho_A$ has full rank, so we can find its inverse $\brho_A^{-1}$ that is given by $\brho_A^{-1} = \sum_i \abs{\psi_i}^{-2} \proj{e_i}$:
		\begin{eqnarray}
			\tr_A\left[\Bigl(\frac{1}{2}\brho_A^{-1} \ten \I\Bigr) \R_{AB}\right] &=& \tr_A\left[\left( \Bigl(\sum_k \frac{1}{2\abs{\psi_k}^2} \proj{e_k}\Bigr) \ten \I\right) 
			\left(\sum_{ij}\psi_i\psi_j^* \ketbra{e_i}{e_j} \ten \ketbra{f_i}{f_j}\right) 
			\right]\\
			&=& \sum_k \frac{\abs{\psi_k}^2}{2\abs{\psi_k}^2} \ketbra{f_k}{f_k} = \sum_k \frac{1}{2} \ketbra{f_k}{f_k}= \frac{\I}{2}.
		\end{eqnarray}
		Thus, we calculate the characteristic equation as follows: 
		\begin{eqnarray}
			0 &=& \det(\U(\R_{AB}^{T_A} - \L_{AB}^{T_A} - \lambda\I)\U^+)\\
			&=& \det \left(
			\begin{bmatrix} 
				\abs{\psi_0}^2 & 0 & 0 & 0 \\
				0 & 0 & \psi_0^*\psi_1 & 0 \\
				0 & \psi_0\psi_1^* & 0 & 0 \\
				0 & 0 & 0 & \abs{\psi_1}^2 
			\end{bmatrix}
			-
			\begin{bmatrix} 
				\abs{\psi_0}^2 -\frac{1}{2} & 0 & 0 & 0 \\
				0 & 0 & 0 & 0 \\
				0 & 0 & 0 & 0 \\
				0 & 0 & 0 & \abs{\psi_1}^2 -\frac{1}{2}
			\end{bmatrix} 
			-
			\begin{bmatrix} 
				\lambda & 0 & 0 & 0 \\
				0 & \lambda & 0 & 0 \\
				0 & 0 & \lambda & 0 \\
				0 & 0 & 0 & \lambda
			\end{bmatrix} \right)\\
			&=& \det 
			\begin{bmatrix} 
				\frac{1}{2}-\lambda & 0 & 0 & 0 \\
				0 & -\lambda & \psi_0^*\psi_1 & 0 \\
				0 & \psi_0\psi_1^* & -\lambda & 0 \\
				0 & 0 & 0 & \frac{1}{2}-\lambda 
			\end{bmatrix} \\
			&=& \left(\lambda - \frac{1}{2}\right)\left(\lambda - \frac{1}{2}\right)\biggl( \lambda^2 - \abs{\psi_0}^2\abs{\psi_1}^2\biggr).
		\end{eqnarray}
		$\bchi_{\Lambda}$ has four eigenvalues, $\frac{1}{2},\frac{1}{2}, \pm \abs{\psi_0}\abs{\psi_1}$, resulting in $\overrightarrow{f}(\R_{AB}) = \abs{\psi_0}\abs{\psi_1}$, so we have $\overrightarrow{f}(\R_{AB}) = E_{\text{neg}}(\R_{AB})$. Essentially the same procedure works for the swapped PDO $\R'_{BA} \equiv \S\left(\proj{\psi}\right)\S\+$, denoted by $\R'_{BA} = \proj{\psi'}$, where $\ket{\psi'} \equiv \sum_i \psi_i \ket{f_i}\ten \ket{e_i}$. We then gain the same result, namely, $\overrightarrow{f}(\R'_{BA}) = \abs{\psi_0}\abs{\psi_1}$.  
        Together with the fact $\overleftarrow{f}(\R_{AB}) = \overrightarrow{f}(\R'_{BA})$, we have $\overleftarrow{f}(\R_{AB}) = \overrightarrow{f}(\R_{AB})$, leading to $f(\R_{AB}) = \overrightarrow{f}(\R_{AB})$. Thus, we arrive at $f(\R_{AB}) = E_{\text{neg}}(\R_{AB})$.\\
  
        Hence, we conclude that $ f(\R_{AB}) = E_{\text{neg}}(\R_{AB})$ whenever $\R_{AB}$ is pure or its marginals are given by the maximally mixed state.
	\end{proof}
 
\noindent\textbf{Upper bound of entanglement negativity for spatial and temporal PDOs.}
In turn, we present a theorem underlying Result 4 in the main text.
\begin{restatable}{thm}{maxent}
    Let $\R_{AB}$ be a density operator that is temporally compatible, i.e., $\R_{AB} \in \calS \cap \calT$. Then,
            \begin{equation}
                    \max_{\R_{AB}\in \mathcal{S} \cap \mathcal{T}} E_\text{neg}(\R_{AB}) \le \frac{\sqrt{2}-1}{2}. \nonumber
            \end{equation}
    \label{thm:max_ent}
\end{restatable}

	\begin{proof}
		We prove the theorem by showing $\max_{\R_{AB}\in \mathcal{T}} E_\text{neg}(\R_{AB}) \le \frac{\sqrt{2}-1}{2}$. Together with the fact that 
  \begin{eqnarray}
      \max_{\R_{AB}\in \mathcal{S} \cap \mathcal{T}} E_\text{neg}(\R_{AB}) \le \max_{\R_{AB}\in \mathcal{T}} E_\text{neg}(\R_{AB}),
  \end{eqnarray}
we complete the proof.
		
		Let $\R_{AB}$ be a temporal PDO, and $\brho_A$ the first marginal of $\R_{AB}$, i.e., $\brho_A \equiv \tr_B \R_{AB}$. Without loss of generality, suppose $\R_{AB}$ is compatible with the temporally distributed structure where measurement on the system $A$ occurs first before that on the system $B$. Then $\R_{AB}$ can be written as $\R_{AB} = (\id \ten \E) \{\brho_A\ten \frac{\I}{2}, \swap\}$ with its associated physical time evolution, that is, a CPTP map $\E$ \cite{zhao18pdogeometry,horsman17pdoandothers}. Recall that the entanglement negativity $E_\text{neg}$ is defined as $E_\text{neg}(\R_{AB})\equiv \frac{\norm{(T \ten \id)\R_{AB}}_{\tr}-1}{2}$. We will thus prove that $\norm{(T \ten \id)\R_{AB}}_{\tr} \le \sqrt{2}$.
		
		By using the definition of the diamond norm $\norm{\cdot}_\diamond$, we have the following inequality,
		\begin{eqnarray}
			\norm{(T \ten \id)\R_{AB}}_{\tr} &=& \norm{(T \ten \E)\left\{\brho_A\ten \frac{\I}{2}, \swap\right\}}_{\tr} \\
			&\le& \norm{\E}_{\diamond} \cdot \norm{(T \ten \id)\left\{\brho_A\ten \frac{\I}{2}, \swap\right\}}_{\tr} \\
			&\le& \norm{(T \ten \id)\left\{\brho_A\ten \frac{\I}{2}, \swap\right\}}_{\tr}. \label{ineq:b}
		\end{eqnarray}
		The last inequality holds because a CPTP map has unit diamond norm. Now, we represent $\brho_A$ as a matrix in a canonical basis as
		\begin{eqnarray}
			\brho_A = \frac{1}{2}
			\begin{bmatrix} 
				1-p & c \\
				c^* & 1+p 
			\end{bmatrix},
		\end{eqnarray}
		with a positive real number $0\le p\le1$ and a complex number $c$. In the matrix representation, the operator on the right hand side in Inequality (\ref{ineq:b}) is represented as 
		\begin{eqnarray}
			(T \ten \id)\left\{\brho_A\ten \frac{\I}{2}, \swap\right\} = \frac{1}{4}
			\begin{bmatrix} 
				2(1-p) & c & c^* & 2 \\
				c^* & 0 & 0 & c^* \\
				c & 0 & 0 & c \\
				2 & c & c^* & 2(1+p) 
			\end{bmatrix}.
		\end{eqnarray}
		After solving the characteristic equation of $(T \ten \id)\{\brho_A\ten \frac{\I}{2},\swap\}$, we obtain its eigenvalues; $0,0,\frac{1}{2}(1\pm\sqrt{1+|p|^2+|c|^2})$:
		\begin{eqnarray}
			0 &=& \det((T \ten \id)\left\{\brho_A\ten \frac{\I}{2},\swap\right\} - \lambda \I) \\
			   &=& \det(\frac{1}{4}
			\begin{bmatrix} 
				2(1-p)-4\lambda & c & c^* & 2 \\
				c^* & -4\lambda & 0 & c^* \\
				c & 0 & -4\lambda & c \\
				2 & c & c^* & 2(1+p)-4\lambda 
			\end{bmatrix}) \\
				&=& \det \frac{1}{4}
			\begin{bmatrix} 
				2(1-p)-4\lambda & c & c^* & 2 \\
				c^* & -4\lambda & 0 & c^* \\
				c & 0 & -4\lambda & c \\
				2p+4\lambda & 0 & 0 & 2p-4\lambda 
			\end{bmatrix}\\
				&=& -\left(\frac{1}{4}\right)^4(2p+4\lambda)
			\det\begin{bmatrix}
				c & c^* & 2 \\
				-4\lambda & 0 & c^* \\
				0 & -4\lambda & c
			\end{bmatrix}
			+\left(\frac{1}{4}\right)^4(2p-4\lambda)
			\det\begin{bmatrix} 
				2(1-p)-4\lambda & c & c^*\\
				c^* & -4\lambda & 0 \\
				c & 0 & -4\lambda
			\end{bmatrix} \\
				&=& \lambda^2 \left(\lambda^2-\lambda-\frac{1}{4}\left(p^2+|c|^2\right)\right) \\
			\Leftrightarrow \lambda &=& 0,\frac{1}{2}\left(1\pm\sqrt{1+p^2+|c|^2}\right). 
		\end{eqnarray}
		
		By the positivity of the quantum state $\brho_A$, we have $0 \le |p|^2+|c|^2 \le 1$, which leads to 
		\begin{eqnarray}
			\norm{(T \ten \id)\left\{\brho_A\ten \frac{\I}{2}, \swap\right\}}_{\tr} &=&  \frac{1}{2}\abs{1+\sqrt{1+|p|^2+|c|^2}} + \frac{1}{2}\abs{1-\sqrt{1+|p|^2+|c|^2}} \\
			&=& \frac{1}{2}\left(1+\sqrt{1+|p|^2+|c|^2}\right) - \frac{1}{2}\left(1-\sqrt{1+|p|^2+|c|^2}\right) \\
			&=& \sqrt{1+|p|^2+|c|^2} \\
			&\le& \sqrt{1+1} = \sqrt{2}.
		\end{eqnarray}
		It follows that $E_\text{neg}(\R_{AB})\le \frac{\sqrt{2}-1}{2}$ for any temporal $\R_{AB}$.\\
  
  Hence, we conclude that $\max_{\R\in \mathcal{S} \cap \mathcal{T}} E_\text{neg}(\R_{AB}) \le \frac{\sqrt{2}-1}{2}$.
	\end{proof}

The positivity of PDOs does not change under local unitary operations, and thus aspatiality of PDOs is invariant under such operations. Together with this fact, \cref{pro:invariant} implies that spatial-temporal compatibility is invariant under local unitary operations.

\begin{restatable}{pro}{inv}
    Atemporality is invariant under any local unitary operations acting on a PDO. \label{pro:invariant}
\end{restatable}

\begin{proof}
    Let $\R_{AB}$ be a PDO. By \cref{lem:exist}, $\R_{AB}$ can be expressed with its forward pseudo-channel $\overrightarrow{\Lambda}$ and the associate initial state $\brho_A = \tr_B \R_{AB}$ as 
    \begin{equation}
        \R_{AB} = \left(\id \ten \overrightarrow{\Lambda}\right) \left\{ \brho_A \ten \frac{\I}{2}, \swap \right\}. \label{eq:fpc}
    \end{equation}
    Likewise, it can be expressed with its reverse pseudo-channel $\overleftarrow{\Lambda}$ and the associate initial state $\btau_B = \tr_A \R_{AB}$ as 
    \begin{equation}
        \R_{AB} = \left(\overleftarrow{\Lambda} \ten \id \right) \left\{  \frac{\I}{2} \ten \btau_B, \swap \right\}.\label{eq:rpc}
    \end{equation}
    
    We denote by $\overrightarrow{\varDelta}(\R_{AB}),\overleftarrow{\varDelta}(\R_{AB})$ the set of forward, reverse pseudo-channels, respectively, compatible with a given PDO $\R_{AB}$. 
    Now we show that (1) $\overrightarrow{\varDelta}\bigl(\left(\mathcal{U} \ten \id \right)(\R_{AB})\bigr) = \left\{ \overrightarrow{\Lambda}\circ\mathcal{U}^\dagger : \overrightarrow{\Lambda} \in \overrightarrow{\varDelta}(\R_{AB}) \right\}$ and (2) $\overleftarrow{\varDelta}\bigl(\left(\mathcal{U} \ten \id \right)(\R_{AB})\bigr) = \left\{ \mathcal{U}\circ\overleftarrow{\Lambda} : \overleftarrow{\Lambda} \in \overleftarrow{\varDelta}(\R_{AB}) \right\}$ for any unitary operation $\mathcal{U}$. Let $\mathcal{U}$ be an arbitrary unitary operation.
    
    (1) For any $\overrightarrow{\Lambda} \in \overrightarrow{\varDelta}(\R_{AB})$, it satisfies Eq.~(\ref{eq:fpc}). Then we can show that forward pseudo-channels for $\left(\mathcal{U} \ten \id \right)(\R_{AB})$ are of the form $\overrightarrow{\Lambda}\circ\mathcal{U}^{\dagger}$: 
    \begin{eqnarray}
        \left(\mathcal{U} \ten \id \right)(\R_{AB}) &=&  \left(\id \ten \overrightarrow{\Lambda}\right) \left\{ \mathcal{U}(\brho_A) \ten \frac{\I}{2}, \left(\mathcal{U} \ten \id \right)(\swap) \right\} \\
            &=& \left(\id \ten \overrightarrow{\Lambda}\right) \left\{ \mathcal{U}(\brho_A) \ten \frac{\I}{2}, \left(\id \ten \mathcal{U}^\dagger \right)(\swap) \right\} \\
            &=& \left(\id \ten \overrightarrow{\Lambda}\circ\mathcal{U}^\dagger\right) \left\{ \mathcal{U}(\brho_A) \ten \frac{\I}{2}, \swap \right\}.
    \end{eqnarray}
    The first equation follows by \cref{lem:exist}, and the second equation holds due to the fact that $(\I \ten \mathbf{M})\ket{\Phi^+} = (\mathbf{M}^T \ten \I)\ket{\Phi^+}$ for any linear operator $\mathbf{M}$, where $\ket{\Phi^+}$ denotes the maximally entangled state. Thus, we have shown that $\overrightarrow{\varDelta}\bigl(\left(\mathcal{U} \ten \id \right)(\R_{AB})\bigr) = \left\{ \overrightarrow{\Lambda}\circ\mathcal{U}^\dagger : \overrightarrow{\Lambda} \in \overrightarrow{\varDelta}(\R_{AB}) \right\}$. We also observe that the relation between $\overrightarrow{\Lambda}$ and $\overrightarrow{\Lambda}\circ\mathcal{U}^\dagger$ is one-to-one correspondence.
    
    (2) For any $\overleftarrow{\Lambda} \in \overleftarrow{\varDelta}(\R_{AB})$, it satisfies Eq.~(\ref{eq:rpc}). Then it can be simply shown that reverse pseudo-channels for $\left(\mathcal{U} \ten \id \right)(\R_{AB})$ are of the form $\mathcal{U}\circ\overleftarrow{\Lambda}$: 
    \begin{eqnarray}
        \left(\mathcal{U} \ten \id \right)(\R_{AB}) &=&  \Bigl(\mathcal{U} \ten \id \Bigr)\left(\overleftarrow{\Lambda} \ten \id \right) \left\{  \frac{\I}{2} \ten \btau_B, \swap \right\}\\
        &=&\left(\mathcal{U}\circ\overleftarrow{\Lambda} \ten \id \right) \left\{  \frac{\I}{2} \ten \btau_B, \swap \right\}.
    \end{eqnarray}
    Thus, we have shown that $\overleftarrow{\varDelta}\bigl(\left(\mathcal{U} \ten \id \right)(\R_{AB})\bigr) = \left\{ \mathcal{U}\circ\overleftarrow{\Lambda} : \overleftarrow{\Lambda} \in \overleftarrow{\varDelta}(\R_{AB}) \right\}$.\\
    
    It is left to show that $f\bigl(\left(\mathcal{U} \ten \id \right)(\R_{AB})\bigr) = f(\R_{AB})$:
    \begin{eqnarray}
        f(\R_{AB}) &\equiv& \min_{\substack{\overrightarrow{\Lambda} \in \overrightarrow{\varDelta}(\R_{AB})\\
        \overleftarrow{\Lambda} \in \overleftarrow{\varDelta}(\R_{AB})}
        } \left( \frac{\norm{\bchi_{\overrightarrow{\Lambda}}}_{\tr} -1}{2}, \frac{\norm{\bchi_{\overleftarrow{\Lambda}}}_{\tr} -1}{2}\right)\\
        &=& \min_{\substack{\overrightarrow{\Lambda} \in \overrightarrow{\varDelta}(\R_{AB})\\
        \overleftarrow{\Lambda} \in \overleftarrow{\varDelta}(\R_{AB})}
        } \left(\frac{\norm{\bchi_{\overrightarrow{\Lambda}\circ\mathcal{U}^\dagger}}_{\tr} -1}{2}, \frac{\norm{\bchi_{\mathcal{U}\circ\overleftarrow{\Lambda}}}_{\tr} -1}{2}\right)\\
        &=&\min_{\substack{\overrightarrow{\Lambda} \in \overrightarrow{\varDelta}\bigl(\left(\mathcal{U} \ten \id \right)(\R_{AB})\bigr)\\
        \overleftarrow{\Lambda} \in \overleftarrow{\varDelta}\bigl(\left(\mathcal{U} \ten \id \right)(\R_{AB})\bigr)}
        } \left( \frac{\norm{\bchi_{\overrightarrow{\Lambda}}}_{\tr} -1}{2}, \frac{\norm{\bchi_{\overleftarrow{\Lambda}}}_{\tr} -1}{2}\right)\\
        &\equiv& f\bigl(\left(\mathcal{U} \ten \id \right)(\R_{AB})\bigr).
    \end{eqnarray}
    The second equation holds due to the fact that $\norm{\cdot}_{\tr}$ is invariant under any unitary operations. The third equation simply follows from above observations. In a similar manner, we can show that $f\bigl(\left(\id \ten \mathcal{V} \right)\R_{AB}\bigr) = f(\R_{AB})$, for any unitary operation $\mathcal{V}$. \\

    Hence, we conclude that atemporality is invariant under any local unitary operations acting on a PDO.
\end{proof}

\begin{restatable}{cor}{entbreaking}
	Let $\R_{AB} \in \mathcal{S} \cap \mathcal{T}$ be a PDO whose first marginal of $\R_{AB}$ is given by the maximally mixed state so that $\tr_B\R_{AB} = \frac{\I}{2}$. The (forward) pseudo-channel for $\R_{AB}$ is an entanglement-breaking channel if and only if $\R_{AB}$ is  separable.
	\label{cor:entbreaking}
\end{restatable}

	\begin{proof}
		As $\tr_B\R_{AB} = \frac{\I}{2}$, we have $\L_{AB} = \zero$. Then it follows that the Choi operator $\bchi_{\Lambda}$ of the (forward) pseudo-channel $\Lambda$ coincides with $\R_{AB}^{T_A}$, i.e., $\bchi_{\Lambda} = \R_{AB}^{T_A}$. By definition of an entanglement-breaking channel, $\bchi_{\Lambda}$ is separable if and only if $\Lambda$ is a entanglement-breaking channel. As $\bchi_{\Lambda} = \R_{AB}^{T_A}$ and we deal with two-level bipartite quantum systems, $\bchi_{\Lambda}$ is separable iff $\R_{AB}$ is separable. Hence, we conclude that $\R_{AB}$ is separable if and only if $\Lambda$ is a entanglement-breaking channel. 
	\end{proof}

    \clearpage
\pagebreak
\subsection{Computing atemporality}\label{app4}

Here we show that a semidefinite programming (SDP) helps computing atemporality effectively, and provide how to formulate a SDP. 

First of all, let us recall the definition of (forward) atemporality:
\begin{equation}
\overrightarrow{f}(\R_{AB})\equiv\min_{ \overrightarrow{\Lambda}}~ \mathcal{N}(\bchi_{\overrightarrow{\Lambda}}),
\end{equation}
where $\mathcal{N}(\X)$ denotes a negativity of $\X$, the sum of the absolute values of all negative eigenvalues of $\X$, i.e., $\mathcal{N}(\X)\equiv \frac{\trnorm{\X}-1}{2}$, and $\bchi_{\overrightarrow{\Lambda}}$ denotes the normalized Choi operator $\bchi_{\overrightarrow{\Lambda}}$ of a pseudo-channel $\overrightarrow{\Lambda}$ that is compatible with $\R_{AB}$. By \cref{lem:exist}, we can derive the following equations that $\overrightarrow{\Lambda}$ must satisfy 
\begin{eqnarray}
        &\overrightarrow{\Lambda}&(\X) + \tr[\brho_A \X] \overrightarrow{\Lambda}(\I) = 2\tr_A[(\X \ten \I) \R_{AB}], \label{eq:app_lambda_X} \\ 
        &\overrightarrow{\Lambda}&(\Y) + \tr[\brho_A \Y] \overrightarrow{\Lambda}(\I) = 2\tr_A[(\Y \ten \I) \R_{AB}], \label{eq:app_lambda_Y}\\
        &\overrightarrow{\Lambda}&(\,\Z) + \tr[\brho_A \,\Z] \overrightarrow{\Lambda}(\I) = 2\tr_A[(\,\Z \ten \I) \R_{AB}], \label{eq:app_lambda_Z}\\
        &\overrightarrow{\Lambda}&(\brho_A) = \tr_{A}\R_{AB}, \label{eq:app_lambda_rho}
\end{eqnarray}
where $\brho_A \equiv \tr_B\R_{AB}$. For cases where $\brho_A$ is rank-deficient, we additionally introduce a Hermitian operator $\btau$ of trace one that satisfy
\begin{eqnarray}
    \overrightarrow{\Lambda}(\I-\brho_A) &=& \btau.
\end{eqnarray}
Note that if $\brho_A$ has full rank, we observe that $\btau$ is fixed because Eqs.~(\ref{eq:app_lambda_X},\ref{eq:app_lambda_Y},\ref{eq:app_lambda_Z},\ref{eq:app_lambda_rho}) already determine $\Lambda(\X),\Lambda(\Y),\Lambda(\Z)$ and $\Lambda(\I)$; otherwise $\btau$ remains to be undetermined (see the proof of \cref{thm:recovermap} for the detailed derivation of these equations and the justification for the introduction of $\btau$). 
We will denote by $\overrightarrow{\Lambda}_{\btau}$ the pseudo-channel with $\btau$ assigned to $\overrightarrow{\Lambda}(\I-\brho_A)$.

We then have the following optimization problem: Given a PDO $\R_{AB}$ and $\brho_A = \R_{AB}$,
\begin{eqnarray}
\begin{aligned}
        \text{minimize } &\mathcal{N}(\bchi_{\overrightarrow{\Lambda}_{\btau}}),\\
    \text{subject to } &\begin{cases} \tr\btau = 1,\\ 
    \btau \text{ is Hermitian},\\
    \overrightarrow{\Lambda}_{\btau}(\X) + \tr[\brho_A \X] \overrightarrow{\Lambda}(\I) = 2\tr_A[(\X \ten \I) \R_{AB}] , \\ 
    \overrightarrow{\Lambda}_{\btau}(\Y) + \tr[\brho_A \Y] \overrightarrow{\Lambda}(\I) = 2\tr_A[(\Y \ten \I) \R_{AB}],\\
    \overrightarrow{\Lambda}_{\btau}(\, \Z) + \tr[\brho_A \,\Z] \overrightarrow{\Lambda}(\I) = 2\tr_A[(\,\Z \ten \I) \R_{AB}] ,\\
    \overrightarrow{\Lambda}_{\btau}(\brho_A) = \tr_{A}\R_{AB},\\
    \overrightarrow{\Lambda}_{\btau}(\I-\brho_A) = \btau
    \end{cases}.
\end{aligned}
\end{eqnarray}

Let us recall Eqs.~(\ref{eq:app_choi_pdo},\ref{eq:app_alt_pdo}) that the Choi operator $\bchi_{\overrightarrow{\Lambda}_{\btau}}$ of $\overrightarrow{\Lambda}_{\btau}$ can be written in terms of 
$\R_{AB}$ and $\btau$ as
\begin{eqnarray}
        \bchi_{\overrightarrow{\Lambda}_{\btau}} = \bigl(\R_{AB}\bigr)^{T_{A}} - \frac{1}{2} \bigl(\tr_B\R_{AB}\bigr)^T \ten \tr_A\R_{AB} + \frac{1}{2}\bigl( \I - \tr_B\R_{AB} \bigr)^T \ten \btau.
\end{eqnarray}
We also note that $\bchi_{\overrightarrow{\Lambda}_{\btau}}$ can be decomposed as $\bchi_{\overrightarrow{\Lambda}_{\btau}} = \mathbf{P} - \mathbf{N}$ with two positive operators $\mathbf{P,N}\ge \zero$, and its negativity $\mathcal{N}(\bchi_{\overrightarrow{\Lambda}_{\btau}})$ is then given by the minimum of $ \frac{1}{2}\bigl(\tr\left[ \mathbf{P} + \mathbf{N} \right]-1\bigr)$ over all possible decomposition $\mathbf{P},\mathbf{N}$.\\

Therefore, we formulate the optimization problem as follows: Given a PDO $\R_{AB}$ and $\brho_A = \tr_B\R_{AB}$,
\begin{eqnarray}
\begin{aligned}
        \text{minimize } & \frac{1}{2}\bigl(\tr\left[ \mathbf{P} + \mathbf{N} \right] - 1\bigr),\\
    \text{subject to } &
    \begin{cases} 
    \mathbf{P} - \mathbf{N} = \bigl(\R_{AB}\bigr)^{T_{A}} - \frac{1}{2} \bigl(\tr_B\R_{AB}\bigr)^T \ten \tr_A\R_{AB} + \frac{1}{2}\bigl( \I - \tr_B\R_{AB} \bigr)^T \ten \btau,\\
    \mathbf{P}, \mathbf{N} \ge \zero,\\
    \tr\btau = 1,\\ 
    \btau \text{ is Hermitian}, \\
    \overrightarrow{\Lambda}_{\btau}(\X) + \tr[\brho_A \X] \overrightarrow{\Lambda}(\I) = 2\tr_A[(\X \ten \I) \R_{AB}] , \\ 
    \overrightarrow{\Lambda}_{\btau}(\Y) + \tr[\brho_A \Y] \overrightarrow{\Lambda}(\I) = 2\tr_A[(\Y \ten \I) \R_{AB}] ,\\
    \overrightarrow{\Lambda}_{\btau}(\,\Z) + \tr[\brho_A \,\Z] \overrightarrow{\Lambda}(\I) = 2\tr_A[(\,\Z \ten \I) \R_{AB}],\\
    \overrightarrow{\Lambda}_{\btau}(\brho_A) = \tr_{A}\R_{AB},\\
    \overrightarrow{\Lambda}_{\btau}(\I-\brho_A) = \btau \end{cases},
\end{aligned}
\end{eqnarray}
and this can be computed efficiently by semidefinite programming.



\end{appendix}

\end{document}